\documentclass[12pt,authoryear,review]{elsarticle}

\makeatletter

  \def\ps@pprintTitle{%
 \let\@oddhead\@empty
 \let\@evenhead\@empty
 \def\@oddfoot{\centerline{\thepage}}%
 \let\@evenfoot\@oddfoot}
\makeatother
\usepackage{amsmath,amssymb,amsfonts}

\usepackage{graphicx}
\usepackage{subfigure}

\usepackage{dsfont}
\usepackage{stackengine}
\usepackage{color}
\usepackage[rightcaption]{sidecap}
\usepackage{epstopdf}

\newtheorem{theorem}{Theorem}
\newtheorem{lemma}[theorem]{Lemma}
\newtheorem{proposition}[theorem]{Proposition}

\newenvironment{proof}[1][Proof]{\begin{trivlist}
\item[\hskip \labelsep {\bfseries #1}]}{\end{trivlist}}

\newcommand{\mcA}{{\mathcal{A}}}
\newcommand{\mcX}{{\mathcal{X}}}
\newcommand{\mcF}{{\mathcal{F}}}
\newcommand{\mcP}{{\mathcal{P}}}

\newcommand{\R}{{\mathds{R}}}

\newcommand{\E}{{\mathbb{E}}}
\renewcommand{\P}{{\mathbb{P}}}

\journal{TBA}

\begin{document}
\begin{frontmatter}

\title{Exploratory Control with Tsallis Entropy for Latent Factor Models
\tnoteref{t1}\\
{\small SIAM J. Financial Mathematics, Forthcoming}}
\tnotetext[t1]{The authors wish to thank participants on the OMI Machine Learning and Quantitative Finance Workshop, Research in Options 2022, AMS-EMS-SMF 2022, IMSI workshop on Advances in Optimal Decision Making under Uncertainty, the Bachelier World Congress 2022, and the Universit\'e Paris 1 Panth\'eon-Sorbonne Finance \& Modeling Webinar.}
\tnotetext[t2]{SJ would like to acknowledge
support from the Natural Sciences and Engineering Research Council of Canada (grants RGPIN-2018-05705 and
RGPAS-2018-522715).}

\author[author1]{Ryan Donnelly}\ead{ryan.f.donnelly@kcl.ac.uk}
\author[author2]{Sebastian Jaimungal}\ead{sebastian.jaimungal@utoronto.ca}
\address[author1] {Department of Mathematics, King's College London, \\ Strand, London, WC2R 2LS, United Kingdom}
\address[author2] {Department of Statistical Sciences, University of Toronto, \\ Toronto, Ontario, Canada, M5G 1Z5}

\date{}

\begin{abstract}
	We study optimal control in models with latent factors where the agent controls the distribution over actions, rather than actions themselves, in both discrete and continuous time. To encourage exploration of the state space, we reward exploration with Tsallis Entropy and derive the optimal distribution over states –- which we prove is $q$-Gaussian distributed with location characterized through the solution of an FBS$\Delta$E and FBSDE in discrete and continuous time, respectively. We discuss the relation between the solutions of the optimal exploration problems and the standard dynamic optimal control solution. Finally, we develop the optimal policy in a model-agnostic setting along the lines of soft $Q$-learning. The approach may be applied in, e.g., developing more robust statistical arbitrage trading strategies.
\end{abstract}

\end{frontmatter}

\section{Introduction}

Reinforcement Learning (RL) is an important technique in optimization which has seen many recent developments. At the heart of RL is a trade-off between acting in a manner which is believed to be optimal and deliberately acting suboptimally with the purpose of better learning about the environment. This is often referred to in RL as exploitation vs. exploration, and it allows algorithms to hone in on an optimal strategy over repeated actions without a priori specifying a model for the system's dynamics.

In some situations it is possible to propose a simplistic model which is to be improved upon through exploration. In this paper, we consider a dynamic optimization problem where the agent desires to explore the state and control space by implementing actions which may not be optimal according to the proposed model. In this situation, repeatedly acting optimally according to the model will result in repeated observations of the system along state trajectories that result from the original model. If instead the agent implements a different action, then she may generate observations in the state and control space that may be used to propose a more precise model of the system, yielding improved future optimization.

Deliberately acting suboptimally comes with opportunity costs, therefore,  exploratory actions should be accompanied with incentives. We represent this incentive through two modifications of a standard dynamic optimal control problem using an approach similar to \cite{wang2020reinforcement} (see also \cite{guo2022entropy} and \cite{firoozi2022exploratory} for extensions in the mean field game context). First, the agent is allowed to randomize the action taken at any point in time, specifying the distribution of this random action. Second, the performance criterion is altered to reward the agent based on a measure of randomness corresponding to their selected distribution. We take this reward to be the Tsallis entropy, first introduced in \cite{tsallis1988possible} (see also \cite{tsallis2011nonadditive}), of the agent's action distribution which generalizes the Shannon entropy. We show that this generalization modifies the optimal distribution over actions to be $q$-Gaussian distributed rather than Gaussian, as is the case with Shannon entropy introduced in  \cite{wang2020reinforcement}.

In addition to applications in RL, exploratory control has been applied in other contexts. For example, in \cite{gao2022state} the randomization of controls is intended to emulate simulated annealing when optimizing non-convex functions, rather than used as a method for generating observations of infrequently visited states. In \cite{tang2022exploratory} the properties of the associated exploratory Hamilton-Jacobi-Bellman (HJB) equations are studied as is the rate of convergence of the optimal exploratory control as exploration goes to zero. Entropy regularization has also been employed in the context of Markov decision processes, as in \cite{neu2017unified} and \cite{geist2019theory}, or other areas related to optimization such as stochastic games as in \cite{savas2019entropy}, \cite{guan2020learning}, and \cite{hao2022entropy}.

Here, we primarily focus on dynamic control in discrete-time, in contrast to those above who study the continuous time case. One primary rationale for working in discrete-time is provide a more precise probabilistic description of the controlled dynamics with randomized actions. Moreover, we incorporate an unobserved latent factor into the environment dynamics to investigate how the exploratory control responds to imperfect observations of the system. We also study the continuous time version of the exploratory problem with latent factors and Tsallis entropy. We proceed to show that the solution to the (discrete- and continuous-time) exploratory control problem can be written in terms of the feedback form of the optimal control in the standard (non-randomized) control problem, and that the solution to the continuous-time problem may act as an approximation to the solution in discrete-time. This is useful in cases where the continuous-time problem has a tractable, implementable, solution, rather than undertaking the additional effort of solving a system of equations numerically in the discrete-time formulation.

After investigating how the exploratory control and its approximations behave in discrete and continuous time, we next illustrate how the Tsallis entropy may be incorporated into the performance criterion for (soft) $Q$-learning \citep{ziebart2010modeling, fox2015taming}. We show that the $Q$-function when evaluated on a distribution of actions are tied to those where the $Q$-function is evaluated on a Dirac mass of actions. 
This approach incentivizes further exploration over the ``greedy'' $Q$-learning algorithms which always select the action that is currently thought to be highest performing, and provides better performance compared to taking random arbitrary actions for the purposes of learning.

The remainder of the paper is organized as follows. In Section \ref{sec:discrete_model}, we propose a general linear-quadratic model with a latent factor, introduce exploratory controls, and solve for the optimal control distribution in feedback form. In Section \ref{sec:numerical_demonstration}, we demonstrate the behaviour of the optimal control distribution, in particular how it depends on the parameters which determine the exploration reward. In Section \ref{sec:continuous_model}, we formulate an analogous continuous-time model which is solved in feedback form. In Section \ref{sec:discrete_continuous_relation}, we show that the continuous-time model acts as an approximation to the discrete time model. Section \ref{sec:Q_learning} discusses the extension to $Q$-learning and Section \ref{sec:conclusion} concludes. Longer proofs are contained in the appendix.

\section{Discrete-Time Model}\label{sec:discrete_model}

We first introduce the dynamics and performance criterion for a standard control problem. Then we propose the model which generalizes the dynamics and performance criterion to allow for exploratory control processes. In the discrete-time formulation we work with a linear-quadratic framework because this makes it easier to establish the relation between the optimal standard control and the optimal exploratory control, as well as the relation between the discrete-time and continuous-time solutions.

\subsection{Standard Control Model}\label{sec:discrete_standard_model}

Let $T>0$ be a fixed time and divide the interval $[0,T]$ into $N$ equally sized subintervals of length $\Delta t = T/N$ with endpoints equal to $t_i = i\,\Delta t$ for $i\in\{0,1,\dots,N\}$. Let $(\Omega, \mathcal{G}, \{\mathcal{G}_n\}_{n=0}^N,\P)$ be a filtered probability space, and let stochastic processes $A=\{A_n\}_{n=0}^{N-1}$ and $M=\{M_n\}_{n=0}^{N}$ be adapted to the filtration $\{\mathcal{G}_n\}_{n=0}^N$, where $M$ is a martingale independent of $A$ with $M_0 = 0$. Additionally, we assume $\E\left[
\sum_{n=0}^{N-1} A_n^2\,\Delta t
\right]<\infty$ and $\E\left[M_n^2\right]<\infty$.

Let $Y=\{Y_n\}_{n=0}^N$ be a process given by
\begin{align}
	Y_n &= Y_0 + \sum_{i=0}^{n-1} A_i\,\Delta t + M_n\,,\label{eqn:unaffected_state}
\end{align}
with $Y_0$ constant. The controlled state process denoted $X^\nu=\{X^\nu_n\}_{n=0}^N$ satisfies
\begin{align}
	X^\nu_n &= Y_n + \sum_{i=0}^{n-1} \gamma_i\,\nu_i\,\Delta t\,,
\end{align}
with $\gamma = \{\gamma_n\}_{n=0}^{N-1}$ is a sequence of constants and where $\nu = \{\nu_n\}_{n=0}^{N-1}$ is the control process. The performance criterion for the standard control model is
\begin{align}
	J^\nu(X_0) &= \E\biggl[-B\,(X^\nu_N)^2 + \sum_{n=0}^{N-1} (D_n\,X^\nu_n\,\nu_n - C_n\,(X^\nu_n)^2 - K_n\,\nu_n^2)\,\Delta t\biggr]\,,\label{eqn:standard_performance}
\end{align}
where $C=\{C_n\}_{n=0}^{N-1}$, $D=\{D_n\}_{n=0}^{N-1}$, and $K=\{K_n\}_{n=0}^{N-1}$ are sequence of constants and $B$ is a constant. We assume that all parameters in the model are chosen such that the solution to equation \eqref{eqn:BSDeltaE_h2} satisfies $K_n - h^{(2)}_{n+1}\,\gamma_n^2\,\Delta t>0$.

We assume that an agent choosing the control $\nu$ is not able to directly observe the processes $A$ or $M$, but is able to observe the process $Y$. Thus, the control $\nu$ must be adapted to the filtration $\mcF$ generated by $Y$, which satisfies $\mcF_n\subset\mathcal{G}_n$ for each $n$. This makes the processes $A$ and $M$ latent factors and, through observations of the trajectory of $Y$, may only be approximately estimated. With the stated partial information setup in mind, we define the agent's dynamic value process
to be
\begin{align}
	H^{(0)}_n &= \sup_{\nu\in\mathcal{A}}\E\biggl[-B\,(X^\nu_N)^2 + \sum_{i=n}^{N-1} (D_i\,X^\nu_i\,\nu_i - C_i\,(X^\nu_i)^2  - K_i\,\nu_i^2)\,\Delta t\,\biggr|\,\mcF_n\biggr]\,,\label{eqn:standard_dynamic_value}
\end{align}
where $\mathcal{A}$ is the set of admissible strategies given by
\begin{align}
	\mathcal{A} &= \biggl\{\nu:\nu_n\in\mcF_n \mbox{ and } \E\biggl[\sum_{n=0}^{N-1}\nu_n^2\biggr]<\infty\biggr\}\,.
\end{align}
To apply standard dynamic programming techniques to obtain the optimal control, it is useful to rewrite the process $Y$ in \eqref{eqn:unaffected_state} in terms of only $\mcF_n$-adapted processes.
To this end, let $\widehat{A}_n = \E[A_n\,|\,\mcF_n]$ and 
$\widehat{M}_n = \sum_{i=0}^{n-1}Y_{i+1}-Y_i-\widehat{A}_i\,\Delta t$. Then we have
\begin{align}
	Y_n &= Y_0 + \sum_{i=0}^{n-1} \widehat{A}_i\,\Delta t + \widehat{M}_n\,,\label{eqn:Y_filter}
\end{align}
and $\widehat{M}$ is easily seen to be a martingale which is $\mcF_n$-adapted. An application of the dynamic programming principle (DPP) yields that the optimal control is given by
\begin{align}
	\nu^*_n &= \frac{\E[h^{(1)}_{n+1}\,|\,\mcF_n]\,\gamma_n + 2\,h^{(2)}_{n+1}\,\widehat{A}_n\,\gamma_n\,\Delta t + (2\,h^{(2)}_{n+1}\,\gamma_n + D_n)\, X^{\nu^*}_n }{2\,(K_n - h^{(2)}_{n+1}\,\gamma_n^2\,\Delta t)}\,,\label{eqn:nu_star}
\end{align}
where the $\mcF_n$-adapted process $h^{(1)}$ and the deterministic sequence $h^{(2)}$ satisfy the backward stochastic difference equation (BS$\Delta$E)
\begin{subequations}
	\begin{align}
	\begin{split}
	    		h^{(1)}_n =&\; \E[h^{(1)}_{n+1}|\mcF_n] + 2\,h^{(2)}_{n+1}\,\widehat{A}_n\,\Delta t 
		\\
		&\;+ \frac{(\E[h^{(1)}_{n+1}|\mcF_n]\,\gamma_n + 2\,h^{(2)}_{n+1}\,\widehat{A}_n\,\gamma_n\,\Delta t)\,(2\,h^{(2)}_{n+1}\,\gamma_n + D_n)}{2\,(K_n - h^{(2)}_{n+1}\,\gamma_n^2\,\Delta t)}\,\Delta t\,,
	\end{split}
    \label{eqn:BSDeltaE_h1}
	\\
	h^{(2)}_n =&\; h^{(2)}_{n+1} - C_n\,\Delta t + \frac{(2\,h^{(2)}_{n+1}\,\gamma_n + D_n)^2}{4\,(K_n - h^{(2)}_{n+1}\,\gamma_n^2\,\Delta t)}\,\Delta t\,,\label{eqn:BSDeltaE_h2}
	\end{align}%
	\label{eqn:BSDeltaE}%
\end{subequations}%
subject to the terminal conditions $h^{(1)}_N = 0$ and $h^{(2)}_N = -B$. The details of this computation are omitted, but they are similar to the derivation of the optimal exploratory control which are provided in Appendix \ref{proof:exploratory_discrete}.

\subsection{Control Model with Exploration}\label{sec:discrete_model_explore}

In this section, we modify the dynamics of the standard control problem to allow for exploratory control processes, and we also modify the performance criterion to incorporate an exploration reward. The exploration reward will incentivize  the agent to choose a randomized control at each time point, even when conditioning on all information available up to the time when the control is implemented. The randomization of the implemented control will also require a modification to the information structure of the problem. It is this modification of the information structure to incorporate randomized actions that makes the continuous time counterpart somewhat poorly specified, but the discrete-time problem precisely specified. To this end, let $\{u_n\}_{n=0}^{N-1}$ be a sequence of independent $\mathcal{U}(0,1)$ random variables, also independent from both $A$ and $M$, and denote the agent's control by $\pi = \{\pi_n\}_{n=0}^{N-1}$, $\pi_n|_{\mcF_n}\in\mcP(\R)$ where  $\mcP(\R)$ is the space of probability measures on $\R$. Hence, $\pi$ is a non-negative random field that integrates to unity. The dynamics of the previous section are modified to
\begin{align}
	F^\pi_n(x) &= \int_{-\infty}^x \pi_n(\nu)\,d\nu\,,\label{eqn:F_pi}
	\\
	\nu^\pi_n &= (F^\pi_n)^{-1}(u_n)\,,\label{eqn:nu_pi}
	\\
	X_n^\pi &= Y_n + \sum_{i=0}^{n-1} \gamma_i\,\nu^\pi_i\,\Delta t\,.\label{eqn:relaxed_X}
\end{align}
We extend the filtration that the agent adapts her strategy to so that she may observe the realizations of her past control. Specifically, we denote the $\sigma$-algebra $\mcF^u_n=\sigma\{(Y_k,u_i):k\leq n, i<n\}$, and the set of admissible exploratory controls $\mathcal{A}^r$ is defined to be $\mcF^u_n$-adapted random fields such that
\begin{align}
	\int_{-\infty}^\infty \pi_n(\nu)\,d\nu &= 1\,, & \E\left[
	\sum_{n=0}^{N-1}(\nu^\pi_n)^2\,
	\Delta t
	\right] &< \infty\,, & \pi_n(\nu) \geq 0\,.
\end{align}
By defining the $\sigma$-algebra $\mcF^u_n$ such that $u_n\notin\mcF^u_n$ and $u_n\in\mcF^u_{n+1}$, at time $n$ the agent does not know what her control $\nu_n$ will be, but she controls the distribution of this random quantity. In light of Equations \eqref{eqn:F_pi}-\eqref{eqn:nu_pi}, the distribution of $\nu_n$ conditional on $\mcF^u_n$ has density $\pi_n$ and in particular,
\begin{align}
	\E\left[\nu^\pi_n\,|\,\mcF^u_n\right] &= \int_{-\infty}^\infty \nu\,\pi_n(\nu)\,d\nu\,,
	\label{eqn:nu_bar}
	\\
	\E\left[
	(\nu_n^\pi)^2\,|\,\mcF^u_n
	\right] 
	&= \int_{-\infty}^\infty \nu^2\,\pi_n(\nu)\,d\nu\,.\label{eqn:nu_sigma}
\end{align}
Further, conditional on $\mcF^u_{n+1}$ at the next time step, she is able to observe $u_n$ and therefore $\nu^\pi_n$ as her implemented control has become realized.

We now modify the performance criterion from \eqref{eqn:standard_performance} to reward exploration in the choice of control. This reward comes in the form of either the Shannon or Tsallis entropy (see \cite{tsallis1988possible}) of the distribution of her control so that the new performance criterion is
\begin{align}
	J^\pi(X_0) &= \E\biggl[-B\,(X^\pi_N)^2 + \sum_{n=0}^{N-1} (D_n\,X^\pi_n\,\nu^\pi_n - C_n\,(X^\pi_n)^2 - K_n\,(\nu^\pi_n)^2)\,\Delta t + \lambda\, S_q[\pi_n]\,\Delta t\biggr]\,,\label{eqn:explore_performance}
\end{align}
where the parameter $\lambda>0$ controls the magnitude of the reward for exploration, and the quantity $S_q\left[\pi_n\right]$ is given by
\begin{align}\label{eqn:entropy}
	S_q[\pi_n] &= \left\{
	\begin{array}{cc}
	\displaystyle
	\frac{1}{q-1} \biggl(1 - \int_{-\infty}^{\infty} (\pi_n(\nu))^q\,d\nu   \biggr)\,, & q > 1/3, q\neq 1
	\\[1em]
	\displaystyle
	-\int_{-\infty}^{\infty} \pi_n(\nu)\,\log(\pi_n(\nu))\,d\nu\,, & q = 1
	\end{array}
	\right.\,.
\end{align}

For a fixed action distribution $\pi_n$, we have $\lim_{q\rightarrow 1}S_q[\pi_n] = S_1[\pi_n]$ which is why we reduce to the Shannon entropy in that limit. The bound $q>1/3$ is necessary so that the resulting optimal densities have finite variance.

\subsection{Optimal Exploratory Control}

In this section, we obtain the optimal exploratory control $\pi^*$ and discuss the relation between this sequence of conditional densities and the optimal control in the standard framework. We apply the DPP to the dynamic value process in the exploratory model (whose analogous value without exploration is \eqref{eqn:standard_dynamic_value}) given by
\begin{equation}
    \begin{split}
        	H^{(\lambda)}_n 
    	=&\;
    	\sup_{\pi\in\mathcal{A}^r}\E\biggl[-B\,(X^\pi_N)^2 + \sum_{i=n}^{N-1} (D_i\,X^\pi_i\,\nu^\pi_i - C_i\,(X^\pi_i)^2 - K_i\,(\nu^\pi_i)^2)\,\Delta t 
    	\\
    	&\qquad\qquad
    	+ \lambda\, S_q[\pi_i]\,\Delta t
    	\;\biggr|\;\mcF^u_n\biggr]\,.
    \end{split}
	\label{eqn:exploration_dynamic_value}
\end{equation}

\begin{proposition}[Discrete-time optimal exploratory control]\label{prop:discrete_optimal_exploration}
	Let $h^{(1)}$ and $h^{(2)}$ satisfy the BS$\Delta$E in Equation \eqref{eqn:BSDeltaE} and assume $K_n - h^{(2)}_{n+1}\,\gamma_n^2\,\Delta t>0$ for each $n$. Then, the optimal control density is given by
	\begin{align}\label{eqn:pi_star}
		\pi_n^*(\nu) &= 
		\left\{
		\begin{array}{cc}
    		\displaystyle	\biggl(\frac{q-1}{\lambda\,q}\biggr)^{\frac{1}{q-1}}\biggl(\psi_n - \frac{\lambda}{2\,\varsigma_n^2} \,(\nu - \mu_n)^2\biggr)_+^{\frac{1}{q-1}}\,, & q > 1
    		\\[1em]
    		\displaystyle	\frac{1}{\sqrt{2\,\pi\,\varsigma_n^2}}\,\exp\biggl\{-\frac{(\nu - \mu_n)^2}{2\,\varsigma_n^2}\biggr\}\,, & q = 1
    		\\[1em]
    		\displaystyle	\biggl(\frac{1-q}{\lambda\,q}\biggr)^{\frac{1}{q-1}}\biggl(\psi_n + \frac{\lambda}{2\,\varsigma_n^2} \,(\nu - \mu_n)^2\biggr)^{\frac{1}{q-1}}\,, & 1/3 < q < 1
		\end{array}
		\right.
	\end{align}
	where
	\begin{align}
		\mu_n &= 
		\displaystyle
		\frac{\E[h^{(1)}_{n+1}|\mcF_n]\,\gamma_n + 2\,h^{(2)}_{n+1}\,\widehat{A}_n\,\gamma_n\,\Delta t + (2\,h^{(2)}_{n+1}\,\gamma_n + D_n)\,X^{\pi^*}_n }{2\,(K_n - h^{(2)}_{n+1}\,\gamma_n^2\,\Delta t)}\,,\label{eqn:mu_n}
		\\[1em]
		\varsigma_n^2 &= \frac{\lambda}{2\,(K_n - h^{(2)}_{n+1}\,\gamma_n^2\,\Delta t)}\,,
		\\[1em]
		\psi_n &= 
		\left\{
		\begin{array}{cc}
		    \displaystyle
		    \Biggl[\frac{\Gamma(\frac{1}{q-1}+\frac{3}{2})}{\sqrt{\pi}\,\Gamma(\frac{1}{q-1}+1)}\biggl(\frac{\lambda\,q}{q-1}\biggr)^{\frac{1}{q-1}}\sqrt{K_n - h^{(2)}_{n+1}\,\gamma_n^2\,\Delta t}\Biggr]^{\frac{1}{\frac{1}{q-1}+\frac{1}{2}}}\,, & q > 1
		    \\[1em]
		    \displaystyle
			\Biggl[\frac{\Gamma(\frac{1}{1-q})}{\sqrt{\pi}\,\Gamma(\frac{1}{1-q}-\frac{1}{2})}\biggl(\frac{\lambda\,q}{1-q}\biggr)^{\frac{1}{q-1}}\sqrt{K_n - h^{(2)}_{n+1}\,\gamma_n^2\,\Delta t}\Biggr]^{\frac{1}{\frac{1}{q-1}+\frac{1}{2}}}\,, & 1/3 < q < 1
		\end{array}
		\right.
	\end{align}
\end{proposition}
\begin{proof}
	See Appendix \ref{proof:exploratory_discrete}.
\end{proof}

The optimal exploratory control in Proposition \ref{prop:discrete_optimal_exploration} includes a location parameter denoted by $\mu_n$. This location parameter has identical feedback form to the optimal control in the standard framework given in Equation \eqref{eqn:nu_star}. Thus, as the $q$-Gaussian distribution is symmetric, at each point in time, the exploratory control is chosen such that it is centered around what would be optimal in the standard approach. When $q=1$, the optimal density is Gaussian, the variance of the random control is directly proportional to the exploration reward parameter $\lambda$, which is expected because there is more incentive to deviate farther from the standard optimal control. When $q\neq 1$ the densities come from the class of $q$-Gaussian distributions which have distinctly different behaviour depending on whether $q$ is greater than or less than $1$. When $q>1$, the density has compact support centered at $\mu_n$ so that the density is only positive for $|\nu-\mu_n|<\sqrt{\psi_n/(K_n - h^{(2)}_{n+1}\,\gamma_n^2\,\Delta t)}$, whereas with $q<1$ the optimal distribution has unbounded support with fat tails. For $q\neq 1$ it is difficult to see directly how the variance of the random controls depends on the reward parameter $\lambda$ and the entropy parameter $q$, but these are plotted in Figure \ref{fig:densities}.

\begin{figure}
	\begin{center}
		{\includegraphics[trim=140 240 140 240, scale=0.4]{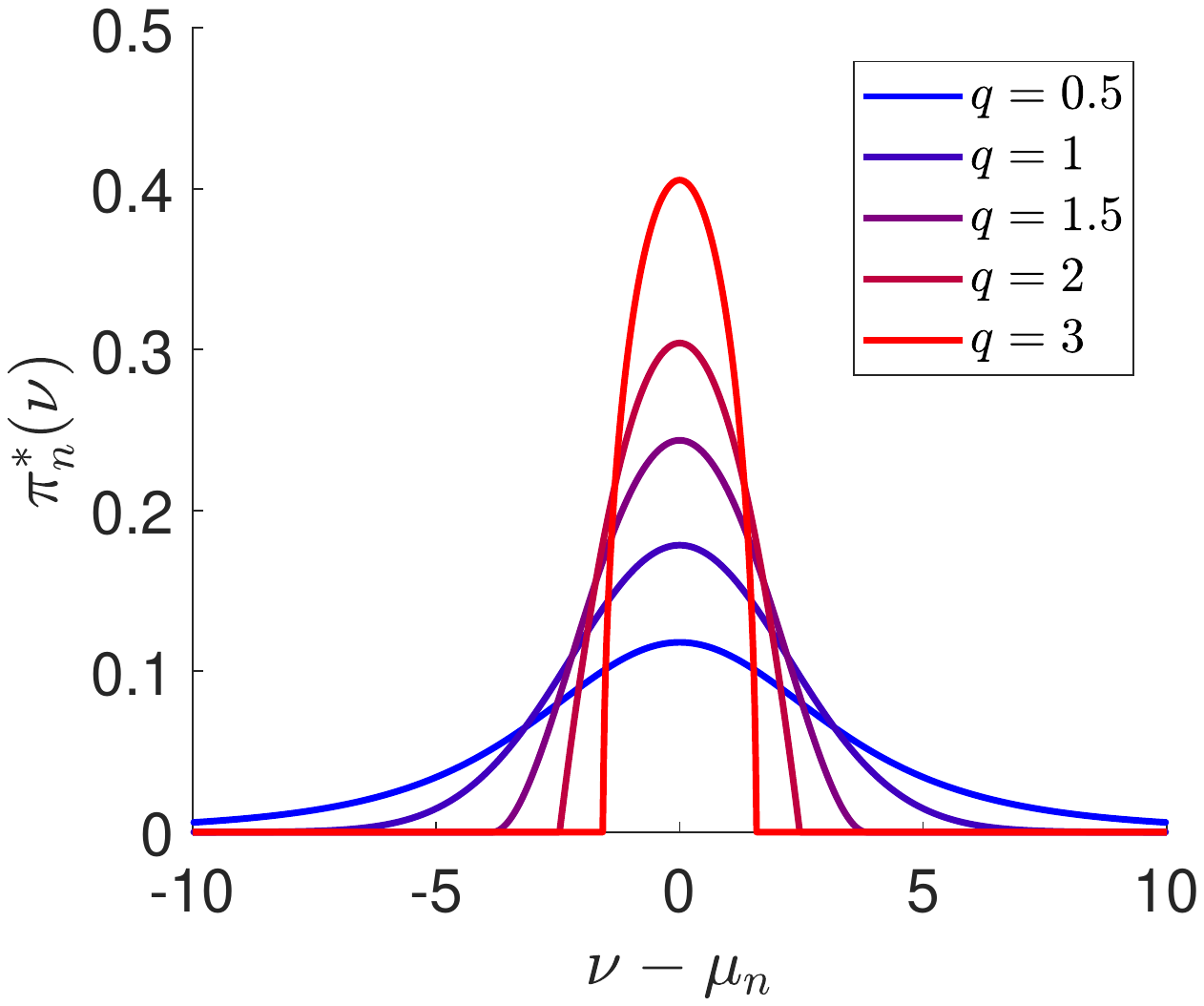}}\hspace{8mm}
		{\includegraphics[trim=140 240 140 240, scale=0.4]{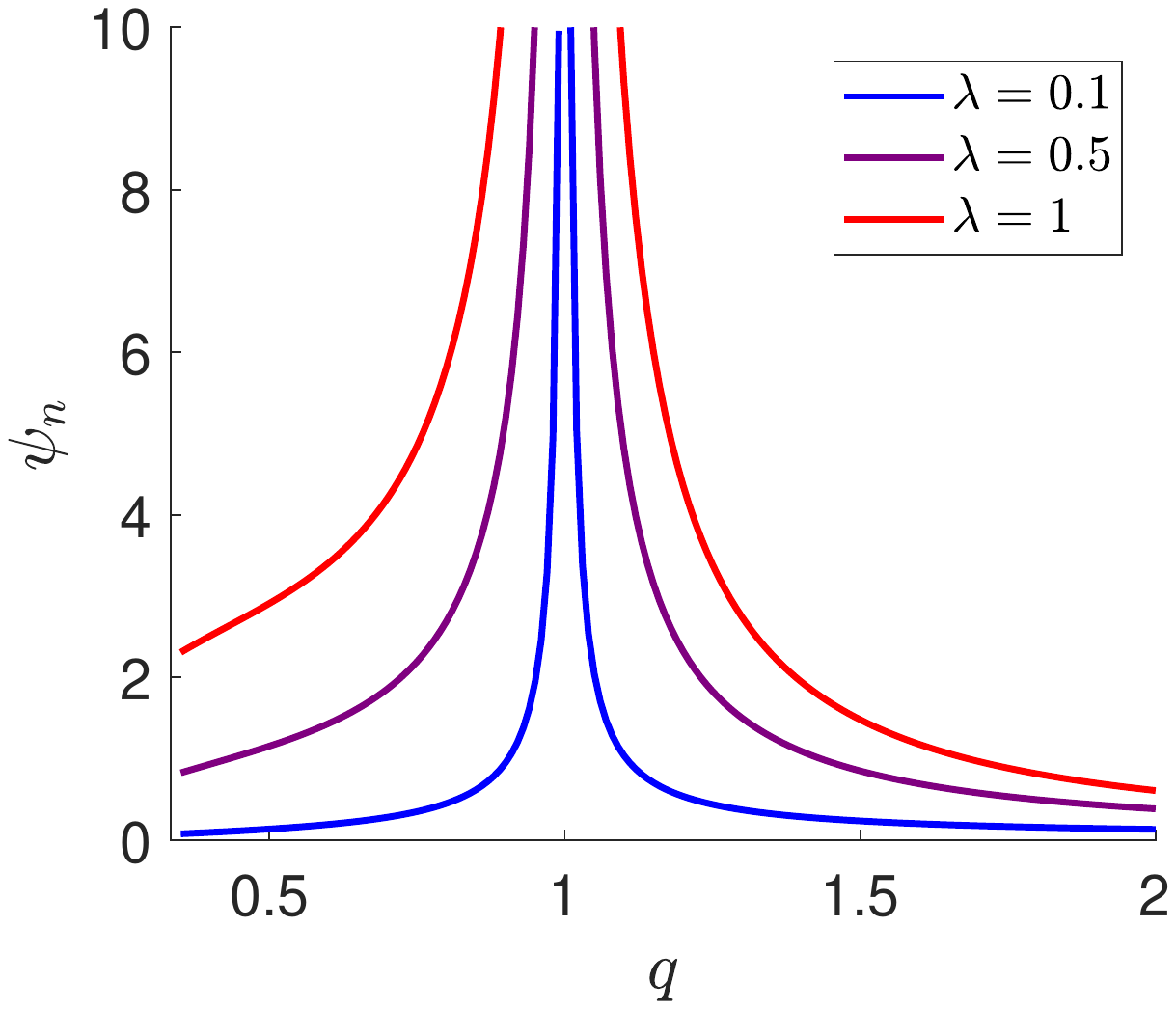}}\hspace{8mm}
		{\includegraphics[trim=140 240 140 240, scale=0.4]{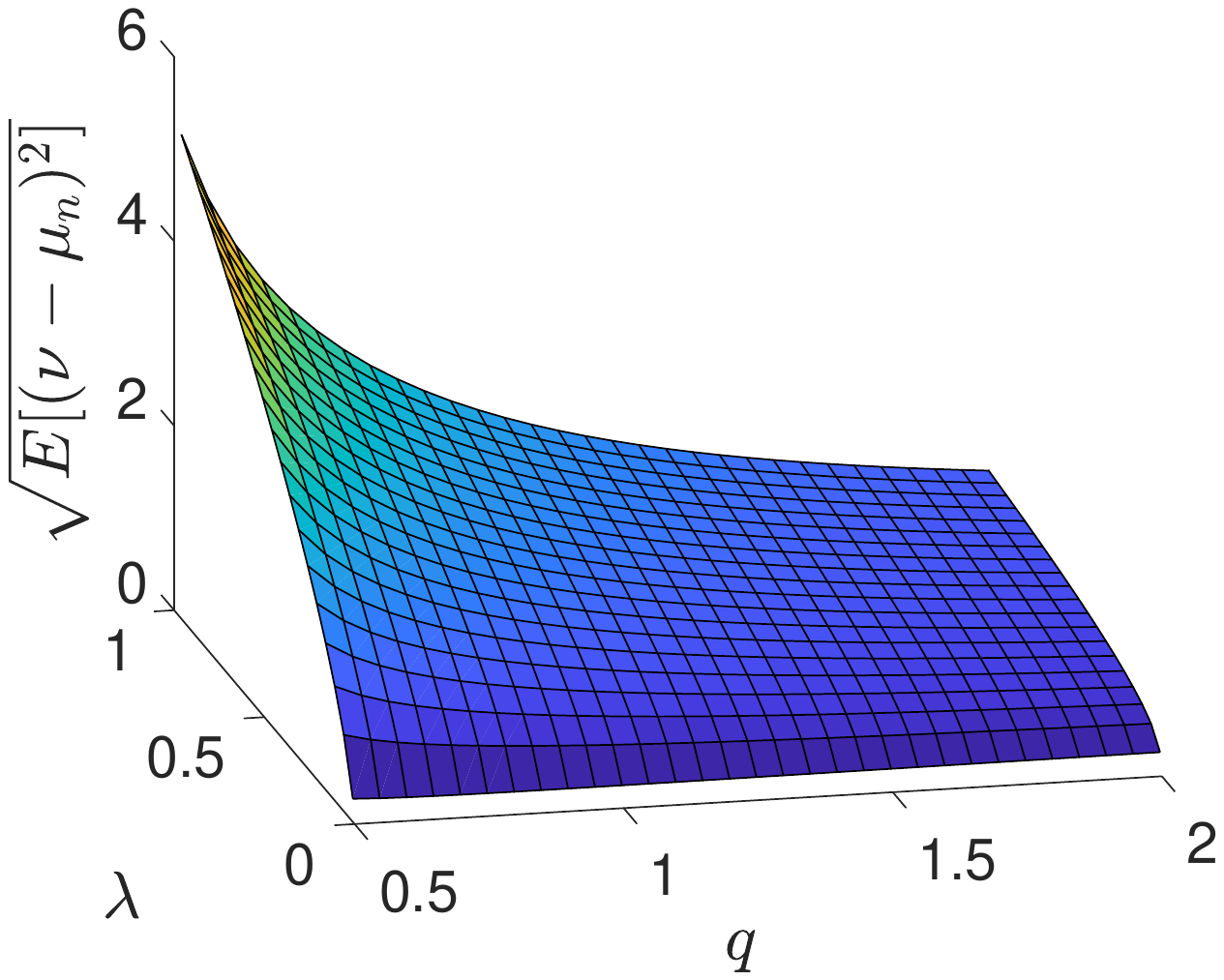}}
	\end{center}
	~\vspace{-1em}
	\caption{The left panel shows the centered optimal density from \eqref{eqn:pi_star} for various values of $q$ with $\lambda = 1$. When $q=1$ the density is Gaussian, otherwise it has bounded support depending on the parameter $\psi_n$. The middle panel shows the parameter $\psi_n$ as a function of $q$ for various values of $\lambda$. The right panel shows the variance corresponding to the density in \eqref{eqn:pi_star}. We set the quantity $K_n - h^{(2)}_{n+1}\,\gamma_n^2\,\Delta t= 0.1$ in all cases.  \label{fig:densities}}
\end{figure}

In Figure \ref{fig:densities} we show examples of the density in \eqref{eqn:pi_star} for various values of $q$. The shrinking support as $q$ increases with $q>1$ is clear, and stems from  $\psi_n$ being decreasing in $q$ in this domain (shown in the middle panel). The right panel  shows the variance of the optimal density as $q$ and $\lambda$ vary. As expected, for all values of $q$ the variance is increasing with respect to $\lambda$. This is due to the higher reward for exploration, however, the dependence is sub-linear and less pronounced when $q$ is large. The variance also decreases with increasing $q$, and is to be expected, because as $q$ increases the distribution becomes less heavy tailed and eventually (for $q>1$) has bounded support with shrinking support as $q$ increases.

\section{Discrete-Time Numerical Experiments}\label{sec:numerical_demonstration}

In this section, we demonstrate, through simulations,\footnote{We use the method given in \cite{thistleton2007generalized} to generate $q$-Gaussian deviates.} how the optimal control distribution behaves and investigate how it depends on the exploration parameter $\lambda$ and the Tsallis entropy parameter $q$. Doing so requires specifying the dynamics of the latent factor process $A$ and unaffected state process $Y$. To keep computations tractable, to obtain the solution to the BS$\Delta$E in \eqref{eqn:BSDeltaE_h1}, we specify the dynamics of $Y$ and $A$ to be
\begin{align}
	Y_n &= Y_0 + \sum_{i=0}^{n-1} A_i\,\Delta t + \sigma\,W^{(1)}_{t_n}\,,\label{eqn:Y_discrete}\\
	A_{n+1} &= e^{-\kappa\,\Delta t}\,A_{n} + \eta\,\Delta W^{(2)}_{t_n}\,,\label{eqn:A_discrete}
\end{align}
where $W^{(1)}$ and $W^{(2)}$ are independent Brownian motions (sampled at the discrete times $\{t_0,\dots,t_N\}$) and $\Delta W^{(2)}_{t_n}:=W^{(2)}_{t_n}-W^{(2)}_{t_{n-1}}$, and independent of $A_0\stackrel{\P}{\sim}\mathcal N(0,\Sigma_0)$. Hence, $A$ may be viewed as a discretely sampled Ornstein-Uhlenbeck (OU) process and $Y$ a discretely observed diffusion with stochastic drift given by $A$. With this specification of $A$ and $Y$, standard filtering results (see, e.g., Theorem 13.4 in \cite{liptser2013statistics2}) show that the filtered process $\widehat{A}$ is given by
\begin{align}
	\widehat{A}_{n+1} &= e^{-\kappa\,\Delta t}\,\widehat{A}_n + \frac{e^{-\kappa\,\Delta t}\,\Sigma_n}{\sigma^2 + \Sigma_n\,\Delta t}\,(Y_{n+1} - Y_n - \widehat{A}_n\,\Delta t)\,, & \widehat{A}_0 &= \E[A_0]\,,\label{eqn:A_hat}
\end{align}
where the deterministic sequence $\Sigma$ represents the estimation error of $\widehat{A}$ and is given recursively by
\begin{align*}
	\Sigma_{n+1} &= e^{-2\,\kappa\,\Delta t}\,\Sigma_n + \eta^2\,\Delta t - \frac{e^{-2\,\kappa\,\Delta t}\,\Sigma_n^2\,\Delta t}{\sigma^2 + \Sigma_n\,\Delta t}\,, \qquad n=0,1,\dots, N-1.
\end{align*}
Linearity of the dynamics of $\widehat{A}$ allow us to represent $h^{(1)}$, the solution to \eqref{eqn:BSDeltaE_h1}, in terms of the solution to a deterministic difference equation.

\begin{proposition}\label{prop:h1_solution}
	Let $\widehat{A}$ be given by \eqref{eqn:A_hat} and let $h^{(2)}$ be the solution to \eqref{eqn:BSDeltaE_h2}. Then the solution to \eqref{eqn:BSDeltaE_h1} is given by
	\begin{align}
		h^{(1)}_n &= \phi_n\,\widehat{A}_n\,,\label{eqn:h1_ansatz}
	\end{align}
	where $\phi$ is the solution to the backwards difference equation
	\begin{align}
		\phi_n &= a_n\,\Delta t + (1 + b_n\,\Delta t)\,e^{-\kappa\,\Delta t}\,\phi_{n+1}\,, & \phi_N &= 0\,,\label{eqn:phi}
	\end{align}
	with
	\begin{align*}
		a_n &= 2\,h^{(2)}_{n+1} \, (1+b_n\,\Delta t)\,,\\
		b_n &= \frac{\gamma_n\,(2\,h^{(2)}_{n+1}\,\gamma_n + D_n)}{2\,(K_n - h^{(2)}_{n+1}\,\gamma_n^2\,\Delta t)}\,.
	\end{align*}
\end{proposition}
\begin{proof}
	The recursion \eqref{eqn:phi} follows from direct substitution of \eqref{eqn:h1_ansatz} into \eqref{eqn:BSDeltaE_h1} and using $\E[\widehat{A}_{n+1}|\mcF_n] = e^{-\kappa\,\Delta t}\,\widehat{A}_n$ from \eqref{eqn:A_hat}. \qed
\end{proof}

For the dynamics given in \eqref{eqn:Y_discrete} and \eqref{eqn:A_discrete}, computing the optimal exploratory control reduces to solving the backward difference equations \eqref{eqn:BSDeltaE_h2} and \eqref{eqn:phi} which is carried out numerically. Given the solutions $h^{(1)}$ and $h^{(2)}$ to \eqref{eqn:BSDeltaE}, the optimal exploratory control in Proposition \ref{prop:discrete_optimal_exploration} can be implemented and demonstrated. For this purpose we set the sequences $C=\{C_n\}_{n=0}^{N-1}$, $D=\{D_n\}_{n=0}^{N-1}$, $K=\{K_n\}_{n=0}^{N-1}$, and $\gamma=\{\gamma_n\}_{n=0}^{N-1}$ to be constant.

\begin{figure}
	\begin{center}
		\includegraphics[trim=140 240 140 240, scale=0.55]{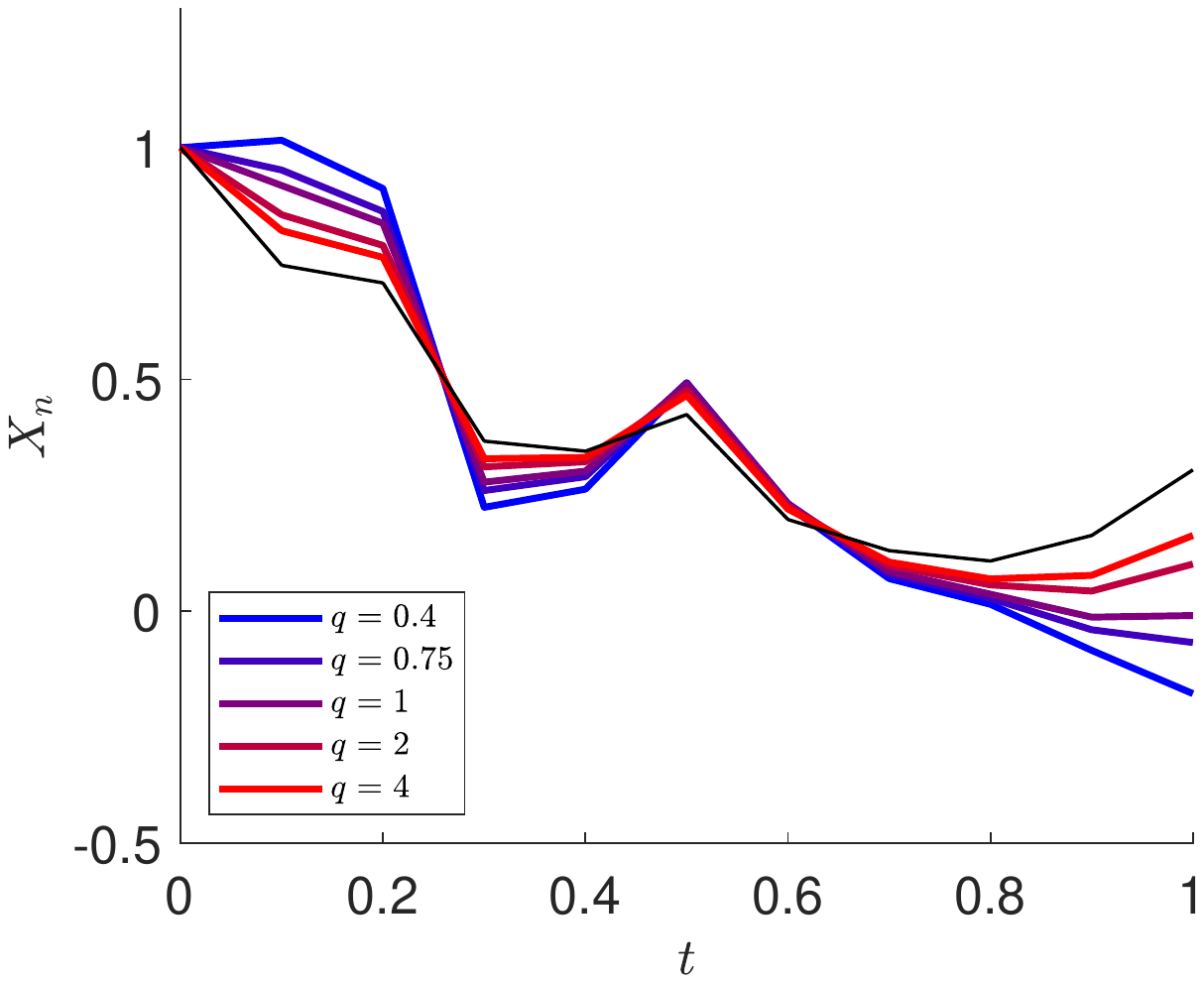}\hspace{10mm}
		\includegraphics[trim=140 240 140 240, scale=0.55]{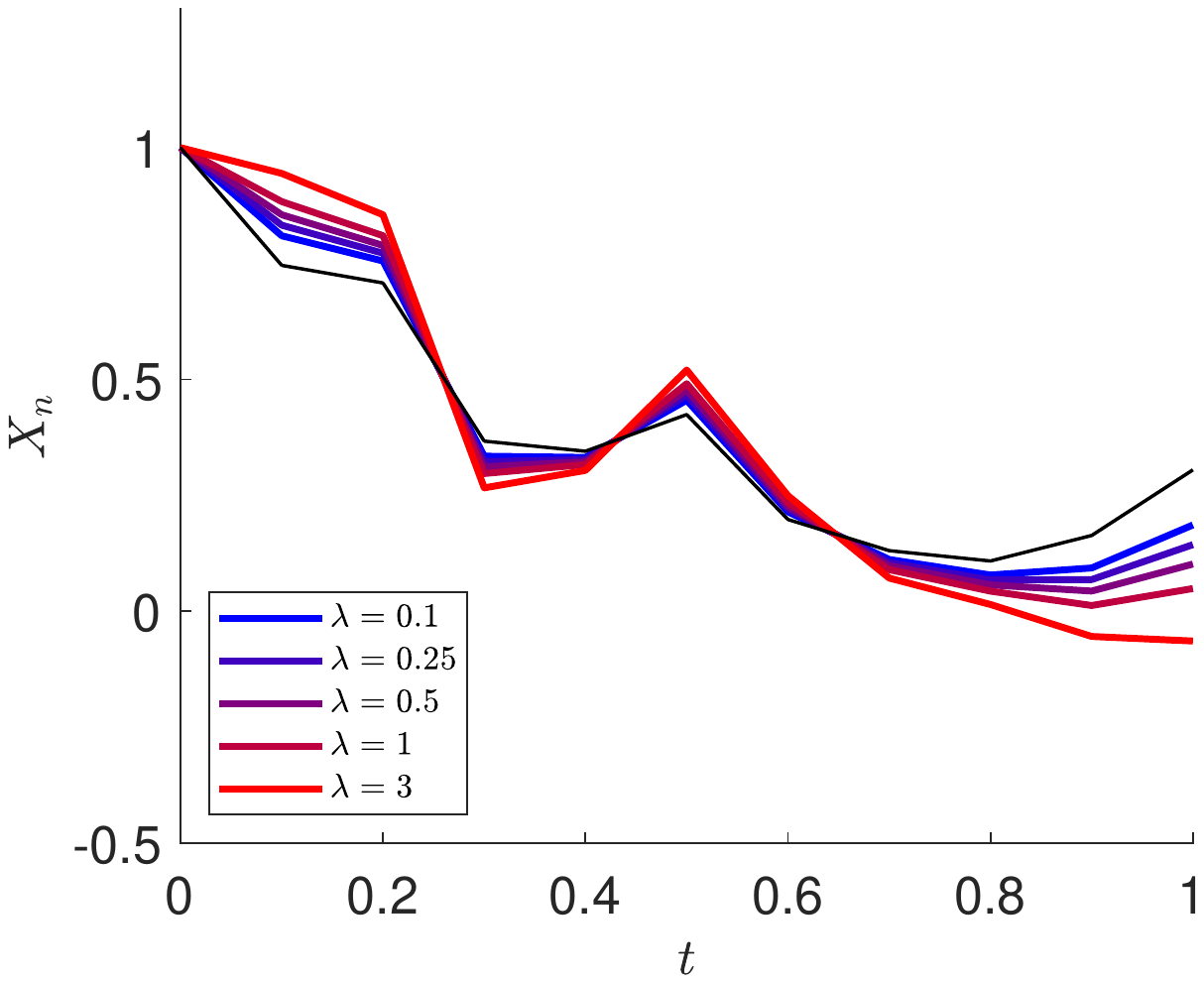}
	\end{center}
	\vspace{-1em}
	\caption{Sample paths for various values of Tsallis entropy parameter $q$ (left) and exploration reward parameter $\lambda$ (right). Other parameter values are $B = 1$, $C=1$, $D=1$, $K=0.1$, $\gamma = 1$, $N=10$, $\sigma = 0.2$, $\kappa = 1$, $\eta = 2$, $\widehat{A}_0 = 0$, $\Sigma_0 = 1$. In the left panel we use $\lambda = 0.5$ and in the right panel we use $q = 2$. The black curve represents the state path using the classical optimal control in \eqref{eqn:nu_star}. \label{fig:MC_discrete_path}}
\end{figure}
In Figure \ref{fig:MC_discrete_path} we show sample paths when implementing the optimal exploratory control. In the left panel we observe that larger values of the Tsallis entropy parameter $q$ typically result in the state path being closer to the one corresponding to the classical optimal control (black curve). This is expected, because as $q$ increases, the $q$-Gaussian distribution becomes more concentrated, as in the discussions surrounding Figure \ref{fig:densities}. The figure also shows that deviation of the paths  from the classical optimal path increases as $\lambda$ increases -- this also expected behaviour due to the increased reward for exploration and deviating from the classical optimizer.

\section{Continuous-Time Model}\label{sec:continuous_model}

Similar to Section \ref{sec:discrete_model}, but now in continuous time, we present the dynamics and performance criterion for a standard control problem so that when we solve the problem with exploration we may directly compare the solutions. 

\subsection{Standard Control Model}\label{sec:continuous_standard_model}

Let $T>0$ be fixed, let $(\Omega, \mathcal{G}, \{\mathcal{G}_t\}_{t\in[0,T]}, \P)$ be a filtered probability space, and let $A=\{A_t\}_{t\in[0,T]}$ and $W^{(1)}=\{W^{(1)}_t\}_{t\in[0,T]}$ be adapted to the filtration $\{\mathcal{G}_t\}_{t\in[0,T]}$, where $W^{(1)}$ is a Brownian motion independent of $A$ and with $\E[\int_0^T A_t^2\,dt]<\infty$.

The unaffected state process is denoted by $Y=\{Y_t\}_{t\in[0,T]}$ and has dynamics given by
\begin{align}
	dY_t &= A_t\,dt + \sigma\,dW^{(1)}_t\,,\label{eqn:dY}
\end{align}
with $Y_0$ and $\sigma > 0$ constant. The controlled state process is denoted by $X^\nu=\{X^\nu_t\}_{t\in[0,T]}$ and satisfies
\begin{align}
	X_t^\nu &= Y_t + \int_0^t \gamma_u\,\nu_u\,du\,,\label{eqn:cont_std_dX}
\end{align}
where $\gamma=\{\gamma_t\}_{t\in[0,T]}$ is deterministic and continuous and $\nu=\{\nu_t\}_{t\in[0,T]}$ is the control process. The performance criterion in the standard control model is
\begin{align}
	J^\nu(X_0) &= \E\biggl[\;g(X_T^\nu) + \int_0^T 
	\left( 
	D_t\,X_t^\nu\,\nu_t - C_t\,(X_t^\nu)^2 - K_t\,\nu_t^2
	\right)
	dt\;\biggr]\,,
\end{align}
where $C=\{C_t\}_{t\in[0,T]}$, $D=\{D_t\}_{t\in[0,T]}$, and $K=\{K_t\}_{t\in[0,T]}$ are deterministic and continuous such that $K_t>0$ for all $t\in[0,T]$. As in discrete-time, admissible controls are processes that are adapted to the filtration $\{\mcF_t\}_{t\in[0,T]}$ generated by $Y$. 

It proves useful to consider the process $\widehat{A}$ defined by $\widehat{A}_t = \E[\,A_t\,|\,\mcF_t\,]$. It is well known that the innovation process $\sigma\,\widehat{W}_t^{(1)}:=Y_t - \int_0^t\widehat{A}_u\,du$ is a (scaled) $\{\mcF_t\}_{t\in[0,T]}$-adapted Brownian motion  (see for example \cite{fujisaki1972stochastic}), which allows us to write 
\begin{align}
	dY_t &= \widehat{A}_t\,dt + \sigma \, d\widehat{W}^{(1)}_t\,.\label{eqn:cont_std_dY}
\end{align}
Additionally, the process $\widehat{A}$ satisfies an SDE of the form (Theorem 4.1 in \cite{fujisaki1972stochastic})
\begin{align}\label{eqn:dA}
    d\widehat{A}_t &= u_t\,dt + v_t\,d\widehat{W}^{(1)}_t\,,
\end{align}
where $u$ and $v$ are $\mcF$-adapted processes.

In the sequel, we assume (with a slight abuse of notation) that $u$ and $v$ are Markov\footnote{Indeed, there are many cases when this holds, and one is the specific example we consider in Section \ref{sec:discrete_continuous_relation}.}, i.e., $u_t = u(t,\widehat{A}_t)$ and $v_t = v(t,\widehat{A}_t)$ for some functions $u,v:\R_+\times\R\to\R$. This additional assumption allows us to more easily take a Hamilton-Jacobi-Bellman (HJB) approach to classifying the optimal control. With this in mind, the corresponding dynamic value function $H^{(0)}$ satisfies
\begin{align}
	H^{(0)}(t,X_t,\widehat{A}_t) &= \sup_{\nu\in\mathcal{A}}\E\biggl[\;g(X_T^\nu) + \int_t^T
	\left(
	D_u\,X_u^\nu\,\nu_u - C_u\,(X_u^\nu)^2 - K_u\,\nu_u^2
	\right)
	du\,\biggl|\,\mcF_t \;\biggr]\,,\label{eqn:cont_std_dynamic}
\end{align}
where the set of admissible strategies is given by
\begin{align}
	\mathcal{A} &= \biggl\{\nu:\nu_t\in\mcF_t \mbox{ and } \E\biggl[\int_0^T\nu_t^2\,dt\biggr]<\infty\biggr\}\,.
\end{align}
The dynamic value function specified by \eqref{eqn:cont_std_dynamic} along with dynamics given by \eqref{eqn:cont_std_dX}, \eqref{eqn:cont_std_dY}, and \eqref{eqn:dA} has associated HJB equation
\begin{equation}\label{eqn:HJB_standard}
	\begin{split}
		\partial_t H^{(0)} + \mathcal{L}^{\widehat{A}}H^{(0)} + \widehat{A}\,\partial_xH^{(0)} + \tfrac{1}{2}\,\sigma^2\,\partial_{xx}H^{(0)} + \sigma\,v\,\partial_{x\widehat{A}}H^{(0)} \hspace{15mm}\\
        -\, C_t\,x^2 + \sup_{\nu}\biggl\{\gamma_t\,\nu\,\partial_xH + D_t\,x\,\nu - K_t\,\nu^2\biggr\} &= 0\,,\\
		H^{(0)}(T,x,\widehat{A}) &= g(x)\,.
	\end{split}
\end{equation}
where $\mathcal{L}^{\widehat{A}} = u\,\partial_{\widehat{A}} + \frac{1}{2}\,v^2\,\partial_{\widehat{A}\widehat{A}}$. We assume that $H^{(0)}$ is a classical solution to \eqref{eqn:HJB_standard}, and thus the optimal control can be written in feedback form as
\begin{align}
	\nu^*(t,X^{\nu^*}_t,\widehat{A}_t) &= \frac{\gamma_t\,\partial_xH^{(0)}(t,X_t,\widehat{A}_t) + D_t\,X^{\nu^*}_t}{2\,K_t}\,,\label{eqn:continuous_standard_control}
\end{align}
which when substituted back into \eqref{eqn:HJB_standard} yields
\begin{equation}\label{eqn:HJB_standard2}
	\begin{split}
		\partial_t H^{(0)} + \mathcal{L}^{\widehat{A}}H^{(0)} + \widehat{A}\,\partial_xH^{(0)} + \tfrac{1}{2}\,\sigma^2\,\partial_{xx}H^{(0)} + \sigma\,v\,\partial_{x\widehat{A}}H^{(0)} - C_t\,x^2 + \frac{(\gamma_t\,\partial_xH^{(0)} + D_t\,x)^2}{4\,K_t} &= 0\,.
	\end{split}
\end{equation}

\subsection{Control Model with Exploration}\label{sec:continuous_model_explore}

In continuous-time when the control is implemented continuously, there is not a proper notion of a delayed observation of the control which necessitated the introduction of a sequence of independent uniform variables and a modification of the agent's filtration in Section \ref{sec:discrete_model_explore}. Instead, we consider a framework similar to other works with exploratory controls, such as \cite{wang2020reinforcement}, \cite{guo2022entropy}, and \cite{firoozi2022exploratory}. Thus, we replace the controlled process $X^\nu$ in the previous section with its exploratory control counterpart
\begin{align}
	X_t^\pi &= Y_t + \int_0^t \int_{-\infty}^\infty \gamma_u\,\nu\,\pi_u(\nu)\,d\nu\,du\,,\label{eqn:continuous_X}
\end{align}
where $\pi = \{\pi_t\}_{t\in[0,T]}$ is a random field adapted to $\{\mcF_t\}_{t\in[0,T]}$ and acts as the control. To provide an incentive for exploration, the performance criterion is modified to include an entropic reward:
\begin{align*}
	J^\pi(X_0) &= \E\biggl[g(X_T^\pi) + \int_0^T \int_{-\infty}^\infty \left(D_t\,X_t^\pi\,\nu - C_t\,(X_t^\pi)^2 - K_t\,\nu^2\right)
	\,\pi_t(\nu)\,d\nu\,dt + \lambda\int_0^T S_q[\pi_t]\,dt\biggr]\,,
\end{align*}
where $S_q[\pi]$ is given by \eqref{eqn:entropy}. The corresponding dynamic value function, denoted by $H^{(\lambda)}$, satisfies
{\small
\begin{align*}
\begin{split}
	H^{(\lambda)}(t,X_t,\widehat{A}_t) 
	= 
	\sup_{\pi\in\mathcal{A}^r}\E\biggl[
	\;&g(X_T^\pi) + \int_t^T \int_{-\infty}^\infty (D_u\,X_u^\pi\,\nu - C_u\,(X_u^\pi)^2 - K_u\,\nu^2)\,\pi_u(\nu)\,d\nu\,du \\
	&
	+ \lambda\int_t^T S_q[\pi_u]\,du\,\biggl|\,\mcF_t
	\;\biggr]\,,    
\end{split}
\end{align*}
}
where the set of admissible exploratory controls is
\begin{align*}
	\mathcal{A}^r &= \biggl\{\pi:\pi_t\in\mcF_t\,, \, \pi_t\geq 0\,, \, \int_{-\infty}^\infty \pi_t(\nu)\,d\nu = 1\,, \mbox{ and } \E\biggl[\int_0^T\int_{-\infty}^\infty \nu^2\,\pi_t(\nu)\,d\nu\,dt\biggr]<\infty\biggr\}\,.
\end{align*}

The following proposition gives the optimal exploratory control analogous to the discrete time setting of Proposition \ref{prop:discrete_optimal_exploration}, and shows that the exploratory control solution may be written in terms of the standard optimal control given in \eqref{eqn:continuous_standard_control}.
Additionally, we draw a relationship between the exploratory value function $H^{(\lambda)}$ and the standard value function $H^{(0)}$.

\begin{proposition}[Continuous-time optimal exploratory control]\label{prop:continuous_optimal_exploration}
	The dynamic value function $H^{(\lambda)}$ is equal to
	\begin{align}
		H^{(\lambda)}(t,x,\widehat{A}) &= H^{(0)}(t,x,\widehat{A}) - \int_t^T \alpha_q(u)\,du\,,
	\end{align}
	where $H^{(0)}$ is the dynamic value function of the standard control problem and where
	\begin{align}
		\alpha_q(t) &= \left\{\begin{array}{cc}
			-\frac{\lambda}{q-1} + \frac{\psi_t\,(q+1)}{3\,q-1}\,, & q > 1
			\\[1em]
			\lambda\,\log(\sqrt{\frac{\pi\,\lambda}{K_t}})\,, & q = 1
			\\[1em]
			-\frac{\lambda}{q-1} - \frac{\psi_t\,(q+1)}{3\,q-1}\,, & 1/3 < q < 1,
			\\[1em]
		\end{array}
		\right.
	\end{align}
	and
	\begin{align}
		\psi_t &= \left\{\begin{array}{cc}
			\Biggl[\frac{\Gamma(\frac{1}{q-1}+\frac{3}{2})}{\sqrt{\pi}\,\Gamma(\frac{1}{q-1}+1)}\biggl(\frac{\lambda\,q}{q-1}\biggr)^{\frac{1}{q-1}}\sqrt{K_t}\Biggr]^{\frac{1}{\frac{1}{q-1}+\frac{1}{2}}}\,, & q > 1
			\\[1em]
			\Biggl[\frac{\Gamma(\frac{1}{1-q})}{\sqrt{\pi}\,\Gamma(\frac{1}{1-q}-\frac{1}{2})}\biggl(\frac{\lambda\,q}{1-q}\biggr)^{\frac{1}{q-1}}\sqrt{K_t}\Biggr]^{\frac{1}{\frac{1}{q-1}+\frac{1}{2}}}\,, & 1/3 < q < 1.
		\end{array}\right.\,.
	\end{align}
	Furthermore, the optimal control density is given by
	\begin{align}\label{eqn:pi_star_cont}
		\pi_t^*(\nu) &= \left\{\begin{array}{cc}
			\biggl(\frac{q-1}{\lambda\,q}\biggr)^{\frac{1}{q-1}}\biggl(\psi_t - \frac{\lambda}{2\,\varsigma_t^2} \,(\nu - \mu_t)^2\biggr)_+^{\frac{1}{q-1}}\,, & q > 1
			\\[1em]
			\frac{1}{\sqrt{2\,\pi\,\varsigma_t^2}}\,\exp\biggl\{-\frac{(\nu - \mu_t)^2}{2\,\varsigma_t^2}\biggr\}\,, & q = 1
			\\[1em]
			\biggl(\frac{1-q}{\lambda\,q}\biggr)^{\frac{1}{q-1}}\biggl(\psi_t + \frac{\lambda}{2\,\varsigma_t^2} \,(\nu - \mu_t)^2\biggr)^{\frac{1}{q-1}}\,, & 1/3 < q < 1,
			\\
		\end{array}\right.
	\end{align}
	where
	\begin{align}
		\mu_t &= \frac{\gamma_t\,\partial_xH^{(0)}(t,X_t^{\pi^*},\widehat{A}_t) + D_t\,X_t^{\pi^*}}{2\,K_t}\,,\label{eqn:continuous_mu}\\
		\varsigma_t^2 &= \frac{\lambda}{2\,K_t}\,.
	\end{align}
\end{proposition}
\begin{proof}
	See Appendix \ref{proof:exploratory_cont}.
\end{proof}

\section{Discrete and Continuous Relation}\label{sec:discrete_continuous_relation}

In this section we demonstrate two relationships between the discrete- and continuous-time formulations of exploration. The first is the convergence of the discrete-time solution to the continuous-time solution, and the second is that we may use some terms which appear in the continuous-time solution to construct approximations to the discrete time solution, which may be desirable because the continuous-time counterparts are generally easier to evaluate in closed form. At the heart of these relationships is the fact that the discrete-time system \eqref{eqn:BSDeltaE} converges to the solution of a continuous-time analogue, which we now show before proceeding to numerical experiments.

To this end, we specify $g(x) = -B\,x^2$, the same terminal reward as in the discrete-time formulation. This also makes the continuous time solution explicitly tractable. With this terminal payoff the corresponding HJB equation can be further simplified with the ansatz $H^{(\lambda)}(t,x,\widehat{A}) = f_0(t,\widehat{A}) + f_1(t,\widehat{A})\,x + f_2(t)\,x^2$. Substituting this form back into the HJB equation for $H^{(\lambda)}$ (see \eqref{eqn:HJB_pf} from the proof of Proposition \ref{prop:continuous_optimal_exploration}) and grouping terms by powers of $x$ gives the three coupled equations
\begin{subequations}
    \begin{align}
    	\partial_tf_0 + \mathcal{L}^{\widehat{A}}f_0 + \sigma\,v\,\partial_{\widehat{A}}f_1+ \widehat{A}\,f_1 + \sigma^2\,f_2 + \frac{\gamma_t^2\,f_1^2}{4\,K_t} - \alpha_q &= 0\,, & f_0(T,\widehat{A}) &= 0\,,\label{eqn:df0}\\
    	\partial_tf_1 + \mathcal{L}^{\widehat{A}}f_1 + 2\,\widehat{A}\,f_2 + \frac{\gamma_t\,f_1\,(2\,\gamma_t\,f_2 + D_t)}{2\,K_t} &= 0\,, & f_1(T,\widehat{A}) &= 0\,,\label{eqn:df1}\\
    	\partial_tf_2 - C_t + \frac{(2\,\gamma_t\,f_2 + D_t)^2}{4\,K_t} &= 0\,, & f_2(T) &= -B\,.\label{eqn:df2}
    \end{align}
\end{subequations}
Defining the processes $h^{(1)}$ and $h^{(2)}$ by
\begin{align}\label{eqn:h12}
	h^{(1)}_t = f_1(t,\widehat{A}_t)\,,
	\qquad \text{and}  \qquad
	h^{(2)}_t = f_2(t)\,,
\end{align}
we may apply It\^o's formula to the process $h^{(1)}$ using \eqref{eqn:df1} and \eqref{eqn:dA}, and a straightforward derivative to $h^{(2)}$ using \eqref{eqn:df2} to obtain
\begin{equation}\label{eqn:BSDE}
	\begin{split}
		dh^{(1)}_t &= (\partial_tf_1 + \mathcal{L}^{\widehat{A}}f_1)\,dt + v(t,\widehat{A}_t)\,\partial_{\widehat{A}}f_1(t,\widehat{A}_t)\,d\widehat{W}^{(1)}_t\\
        &= -\biggl(2\,\widehat{A}_t\,h_t^{(2)} + \frac{\gamma_t\,h_t^{(1)}\,(2\,\gamma_t\,h_t^{(2)} + D_t)}{2\,K_t}\biggr)\,dt + v(t,\widehat{A}_t)\,\partial_{\widehat{A}}f_1(t,\widehat{A}_t)\,d\widehat{W}^{(1)}_t\,,\\
		dh^{(2)}_t &= -\biggl(\frac{(2\,\gamma_t\,h_t^{(2)} + D_t)^2}{4\,K_t} - C_t\biggr)\,dt\,.
	\end{split}
\end{equation}
with terminal conditions $h^{(1)}_T = 0$ and $h^{(2)}_T = -B$. Inspection of the BS$\Delta$E \eqref{eqn:BSDeltaE} in the formal limit $\Delta t\rightarrow 0$ shows a similarity to the BSDE \eqref{eqn:BSDE}. This form of the value function also reduces \eqref{eqn:continuous_mu} to
\begin{align}
	\mu_t &= \frac{\gamma_t\,h^{(1)}_t + (2\,\gamma_t\,h^{(2)}_t + D_t)\,X_t^{\pi^*}}{2\,K_t}\,.
\end{align}
Note again the formal similarity to \eqref{eqn:mu_n} in the limit $\Delta t\rightarrow 0$. 

Below we establish that for sufficiently small $\Delta t$, the processes $h^{(1)}$ and $h^{(2)}$ in \eqref{eqn:h12} approximate the solution to \eqref{eqn:BSDeltaE}. To this end, denote by $h^{(i,N)}$ the solutions to \eqref{eqn:BSDeltaE} when $\Delta t = T/N$. Additionally, with slight abuse of notation, all discrete-time processes and functions in the model are understood to be evaluated at grid points of the continuous-time analogue (such as $K_n = K_{t_n}$).

\begin{lemma}[Convergence of $h^{(i)}$]\label{lem:convergence}
    Let $\beta_0>0$. Suppose that for each $\beta\in[0,\beta_0)$ there is a solution to
    \begin{align*}
        dh^{(2,\beta)}_t &= -\biggl(\frac{(2\,\gamma_t\,h_t^{(2,\beta)} + D_t)^2}{4\,(K_t - h_t^{(2,\beta)}\,\gamma_t^2\,\beta) } - C_t\biggr)\,dt\,, & h^{(2,\beta)}_T &= -B\,,
    \end{align*}
    such that
    \begin{align*}
         \inf_{\substack{t\in[0,T] \\ \beta\in[0,\beta_0)}}\biggl\{K_t - h_t^{(2,\beta)}\,\gamma_t^2\,\beta\biggr\}>0\,.
    \end{align*}
    Then for every $\epsilon>0$ there exists $N^*$ such that for $N\geq N^*$ 
    \begin{align*}
        \max_{0\leq n\leq N}\mathbb{E}\biggl[(h^{(1)}_{t_n} - h^{(1,N)}_n)^2\biggr] + \max_{0\leq n\leq N}(h^{(2)}_{t_n} - h^{(2,N)}_n)^2 & < \epsilon\,.
    \end{align*}
\end{lemma}
\begin{proof}
	See Appendix \ref{proof:convergence}.
\end{proof}

The implementation of the continuous-time solution in the following sections requires a specification of a process $A$, and subsequently a computation of the projection $\widehat{A}$. With the goal of relating the discrete-time solution to the continuous-time solution, we use the continuous analogue of \eqref{eqn:A_discrete} for the dynamics of $A$, namely a mean-reverting OU process satisfying
\begin{align}
	dA_t &= -\kappa\,A_t\,dt + \eta\,dW^{(2)}_t\,,
\end{align}
where $W^{(2)}$ is a Brownian motion independent of $W^{(1)}$ appearing in \eqref{eqn:dY}, and with $A_0$ having Gaussian distribution with variance $\Sigma_0$. By standard filtering results (see Theorem 12.7 in \cite{liptser2013statistics2}), the dynamics of the filtered process are given by
\begin{align}
	d\widehat{A}_t &= -\kappa\,\widehat{A}_t\,dt + \frac{\Sigma_t}{\sigma^2}\,(dY_t - \widehat{A}_t\,dt)\,, & \widehat{A}_0 &= \E[A_0]\,,
\end{align}
where $\Sigma_t$ represents the estimation error of $\widehat{A}$ and is given by
\begin{equation}
	\begin{split}
		\Sigma_t &= \frac{(1+\alpha^+\,\Sigma_0)\,e^{\xi\,t} - (1+\alpha^-\,\Sigma_0)\,e^{-\xi\,t}}{-\alpha^-\,(1+\alpha^+\,\Sigma_0)\,e^{\xi\,t} + \alpha^+\,(1+\alpha^-\,\Sigma_0)\,e^{-\xi\,t}}\,,
		\\[1em]
		\alpha^\pm &=	\frac{-\sigma\,\kappa \pm \sqrt{\sigma^2\,\kappa^2 + \eta^2}}{\sigma\,\eta^2}\,,
		\\[1em]
		\xi &= \frac{\sqrt{\sigma^2\,\kappa^2 + \eta^2}}{\sigma} \,.
	\end{split}
\end{equation}

With this specification of $A$ in mind, we may also write the solution of \eqref{eqn:BSDE} in closed form when the processes $C$, $D$, $K$, and $\gamma$ are constant. The solution, which can be checked by direct substitution, is given by
\begin{equation}\label{eqn:h_t}
	\begin{split}
		h^{(1)}_t &= \frac{\widehat{A}_t}{\gamma}\,\biggl( \frac{\psi^-\,(2\sqrt{C\,K}-D)\,(e^{-\kappa\,(T-t)} - e^{-\omega\,(T-t)})}{(\omega-\kappa)\,(\psi^-\,e^{-\omega\,(T-t)} + \psi^+\,e^{\,\omega\,(T-t)})} + \frac{\psi^+\,(2\sqrt{C\,K}+D)\,(e^{-\kappa\,(T-t)} - e^{\omega\,(T-t)})}{(\omega+\kappa)\,(\psi^-\,e^{-\omega\,(T-t)} + \psi^+\,e^{\,\omega\,(T-t)})} \biggr)
		\\[1em]
		h^{(2)}_t &= \frac{\sqrt{C\,K}}{\gamma}\,\frac{\psi^-\,e^{-\omega\,(T-t)} - \psi^+\,e^{\,\omega\,(T-t)}}{\psi^-\,e^{-\omega\,(T-t)} + \psi^+\,e^{\,\omega\,(T-t)}} - \frac{D}{2\,\gamma}
		\\[1em]
		\psi^\pm &= \sqrt{2\,C} \pm \sqrt{\frac{2}{K}}\,\gamma\,B \mp \frac{D}{\sqrt{2\,K}}
		\\[1em]
		\omega &= \gamma\,\sqrt{\frac{C}{K}}
	\end{split}
\end{equation}

\subsection{Discrete Convergence}

Here, we discuss the behaviour of the discrete-time solution as $\Delta t\rightarrow 0$ as it relates to the continuous-time solution. First we inspect the continuous-time solution given in Proposition \ref{prop:continuous_optimal_exploration} and the dynamics of $X^{\pi}$ given in \eqref{eqn:continuous_X}, which reduce to
\begin{align}\label{eqn:equiv}
		X_t^{\pi^*} = Y_t + \int_0^t \int_{-\infty}^\infty \gamma_u\,\nu\,\pi^*_u(\nu)\,d\nu\,du
		&=  Y_t + \int_0^t \gamma_u \, \mu_u\,du
		\nonumber\\
		&=  Y_t + \int_0^t \gamma_u \, \nu_u^*\,du = X_t^{\nu^*}\,,
\end{align}
where $\nu^*$ is the optimal classical control given in \eqref{eqn:continuous_standard_control} and $X^{\nu^*}$ is the resulting controlled process. Thus, when implementing the optimal exploratory control in continuous-time, the resulting state process coincides with that of the classical control problem.
\begin{figure}[t!]
	\begin{center}
	    \centering {\small varying $\lambda$}
	    \\
		\includegraphics[trim=140 240 140 240, scale=0.42]{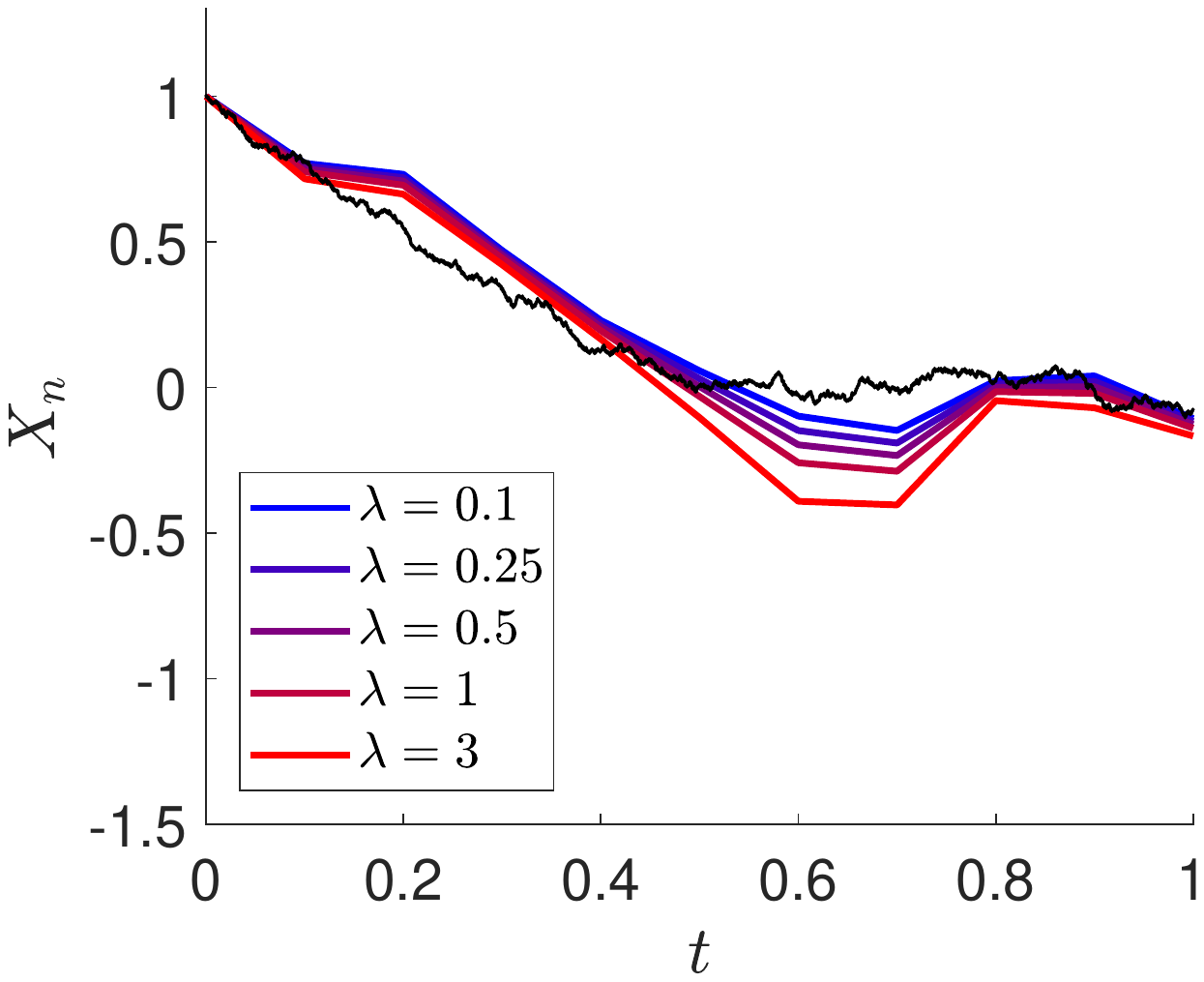}\hspace{8mm}
		\includegraphics[trim=140 240 140 240, scale=0.42]{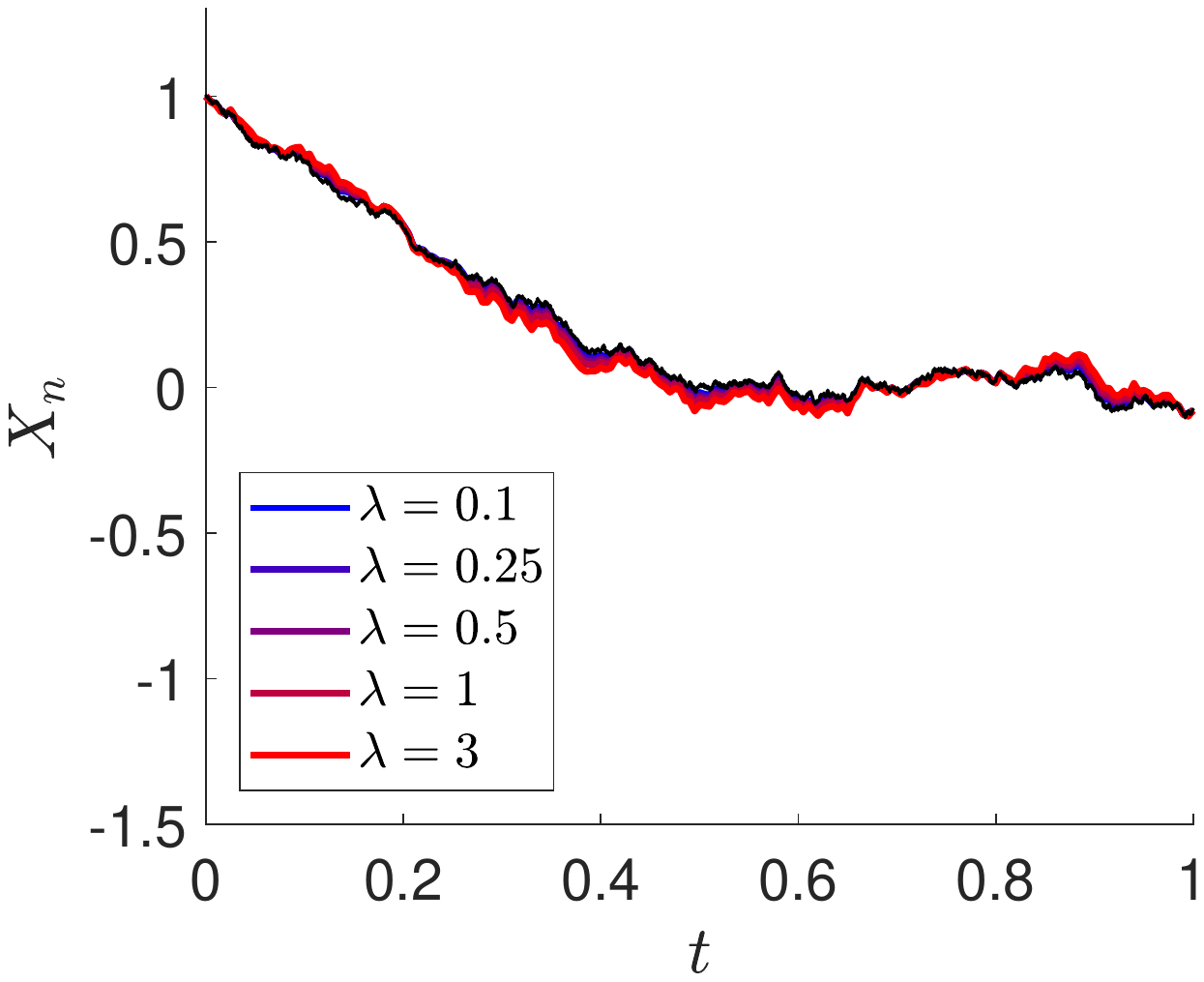}\hspace{8mm}
		\includegraphics[trim=140 240 140 240, scale=0.42]{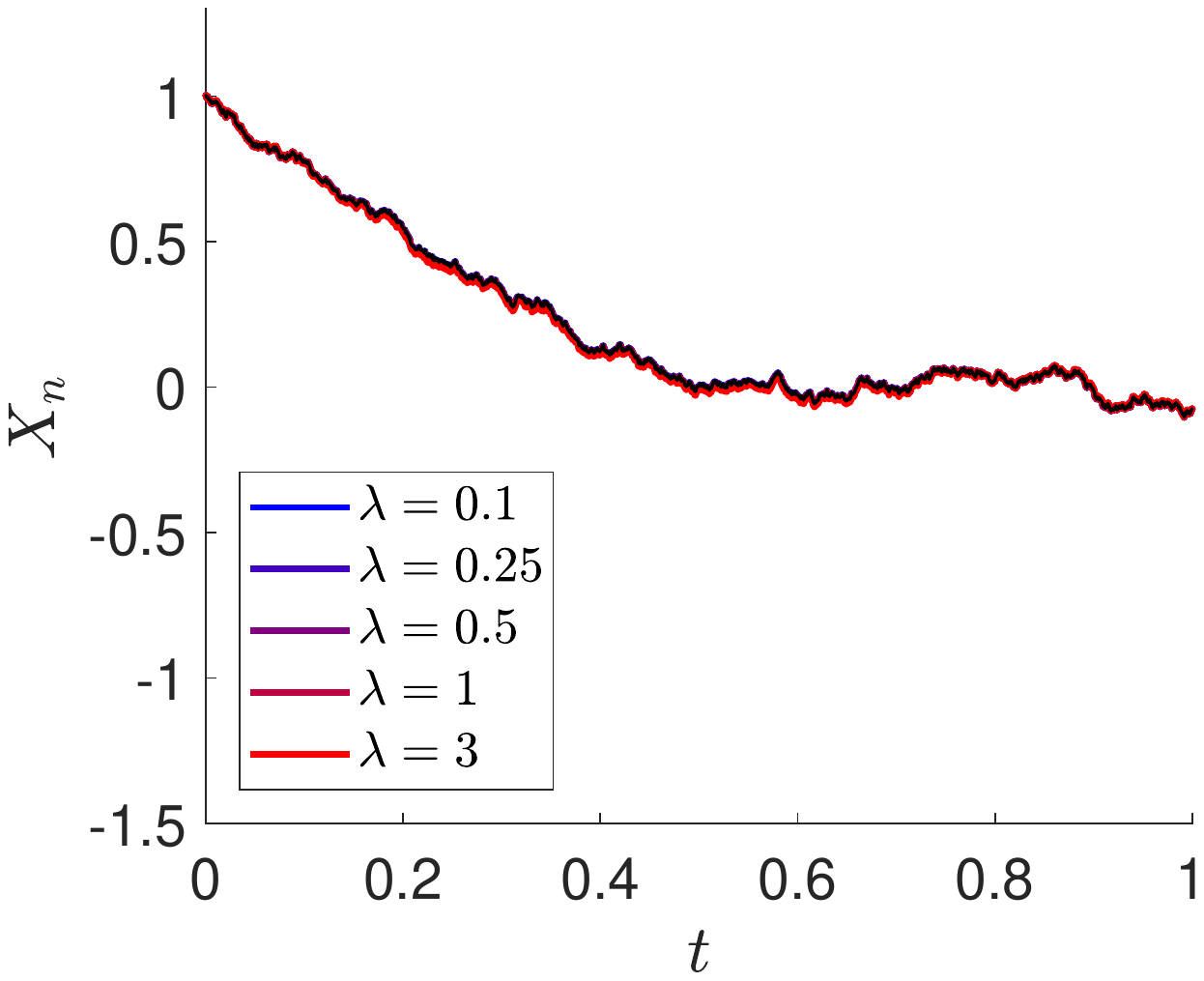}
		\\ 
	    \centering {\small varying $q$}
	    \\		
		\includegraphics[trim=140 240 140 240, scale=0.42]{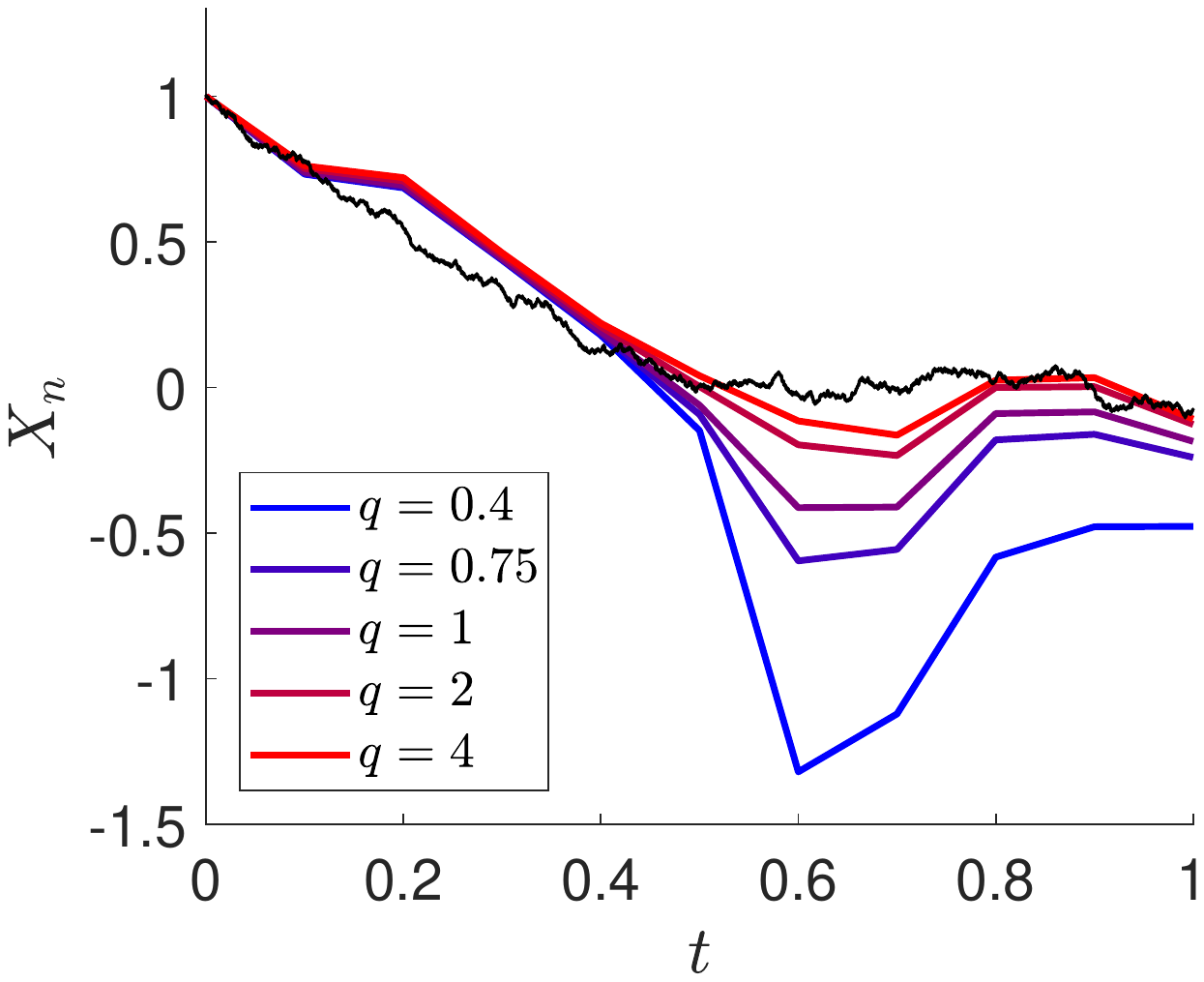}\hspace{8mm}
		\includegraphics[trim=140 240 140 240, scale=0.42]{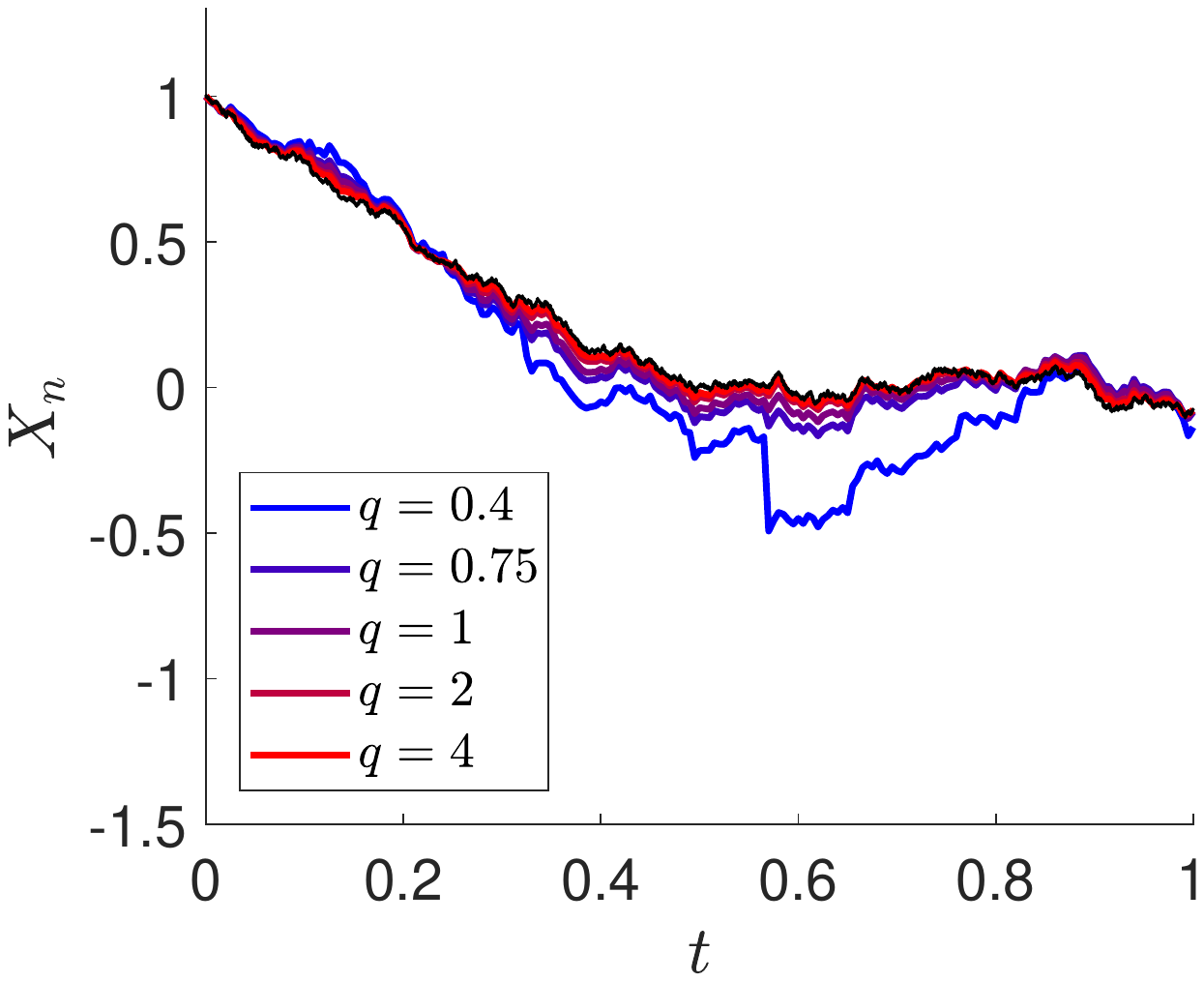}\hspace{8mm}
		\includegraphics[trim=140 240 140 240, scale=0.42]{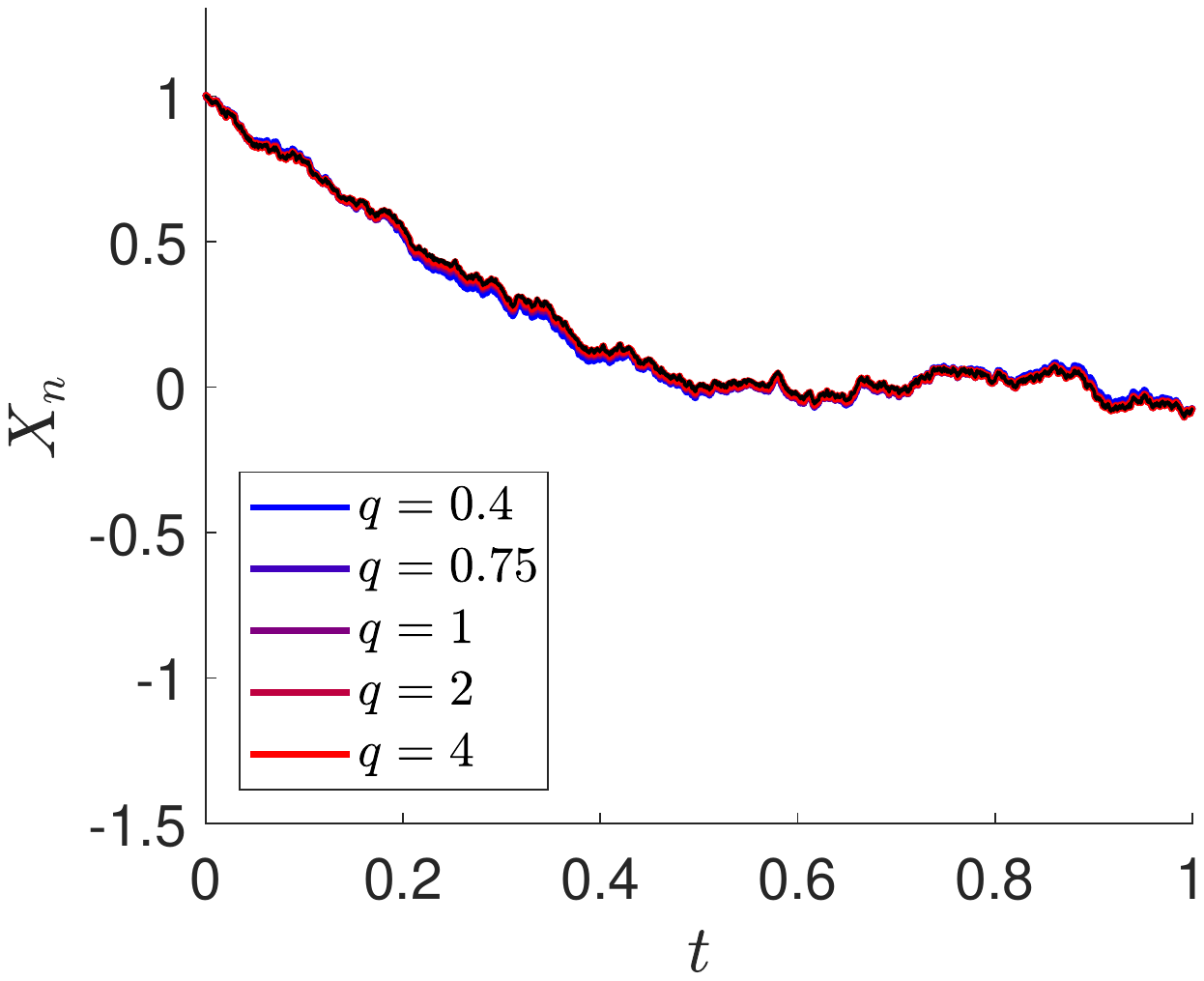}
	\end{center}
\vspace{-1em}
\caption{Convergence of discrete-time paths to continuous-time path. The top row of panels uses $q = 2$ and the bottom row uses $\lambda = 0.5$. The left column has $T/\Delta t = 10$, the middle column $T/\Delta t = 200$, and the right column $T/\Delta t = 10,000$. Other parameter values are $B = 1$, $C=1$, $D=1$, $K=0.1$, $\gamma = 1$, $\sigma = 0.2$, $\kappa = 1$, $\eta = 2$, $\widehat{A}_0 = 0$, $\Sigma_0 = 1$. In each panel, the black curve represents the optimal state path for the continuous-time problem (either the classical or exploratory version which result in the same path). \label{fig:discrete_convergence}}
\end{figure}

In each panel of Figure \ref{fig:discrete_convergence}, the black curve shows the optimal trajectory in continuous-time for one particular outcome of the latent processes $A$ and $W^{(1)}$.
We also show trajectories in discrete-time when the control is drawn from the optimal exploratory distribution in \eqref{eqn:pi_star} for various values of the exploration reward parameter, $\lambda$, and the Tsallis entropy parameter, $q$. Each column of panels corresponds to a different granularity of time discretization. We see that as $\Delta t\rightarrow 0$ the paths of the state process under the optimal exploratory control approach the path of the optimal state process in continuous-time. This effect is essentially due to the law of large numbers applied to the sum of exploratory controls in \eqref{eqn:relaxed_X}. As the number of terms in this sum approaches infinity, the randomness attributed to the sequence of independent uniforms $\{u_n\}_{n=0}^{N-1}$ disappears. We are still left with an element of randomness contributed by the conditional distributions $\{F_n^\pi\}_{n=0}^{N-1}$ which yields a stochastic drift component to the process, as well as the Brownian innovations inherent in the unaffected process $Y$. This figure also shows that for values of the Tsallis entropy parameter $q$ close to the lower bound of $1/3$ the fat-tailed behaviour of the optimal distribution can result in significant deviations from the optimal standard path (blue curve in bottom left and bottom middle panels), but for fixed $q$ this deviation will always disappear as $\Delta t\rightarrow 0$.

\subsection{Discrete Solution Approximation}

In this section we use the solution of \eqref{eqn:BSDE} in continuous-time to act as an approximation of the corresponding processes which appear in the discrete-time optimal exploratory control of Proposition \ref{prop:discrete_optimal_exploration}. The complicated form of \eqref{eqn:BSDeltaE} along with it's discrete nature means numerical methods must be used to solve the backward system, as was done to demonstrate its implementation in Section \ref{sec:numerical_demonstration}. As seen in \eqref{eqn:h_t}, the solution \eqref{eqn:BSDE} may be solved in closed form, thus better lending itself for more efficient implementation.

To this end, let $\tilde{\pi} = \{\pi_n\}_{n=0}^{N-1}$ be a sequence of conditional densities given by \eqref{eqn:pi_star}, except with all occurrences of $h^{(1)}_n$ and $h^{(2)}_n$ replaced with $h^{(1)}_{t_n}$ and $h^{(2)}_{t_n}$ respectively.\footnote{We have abused notation by referring to solutions to both the BS$\Delta$E and BSDE by $h^{(i)}$, but it should be clear from context that $h^{(i)}_{t_n}$ refers to the solution of the BSDE evaluated at time $t_n = n\,\Delta t$, whereas $h^{(i)}_n$ is the solution of the BS$\Delta$E evaluated at the $n^{th}$ step.} In Figure \ref{fig:discrete_approximation} we show paths resulting from the optimal control density along with this approximation. Unsurprisingly, the error between these processes reduces as $\Delta t$ decreases. Moreover, the error is typically smallest towards the endpoints of the time interval $[0,T]$. Both the optimal and approximate $h^{(i)}$ processes have the same terminal condition, and errors in computing an approximation to the optimal control accumulate as we move backwards in time, hence the largest discrepancies in the control process occur near $t=0$. Since the state process begins at the same constant regardless of whether an approximation is being used or not, the error increases from zero rapidly near $t=0$. However, the performance criterion having a terminal reward of the form $-B\,(X^\pi_N)^2$ implies that as time approaches $T$ the controlled state process has a soft target of $0$, regardless of whether the implemented control is optimal or an approximation. As both the optimal process and the approximation have the same target, they tend to come closer together again as the end of the horizon approaches.

\begin{figure}[ht]
	\begin{center}
		\includegraphics[trim=140 240 140 240, scale=0.50]{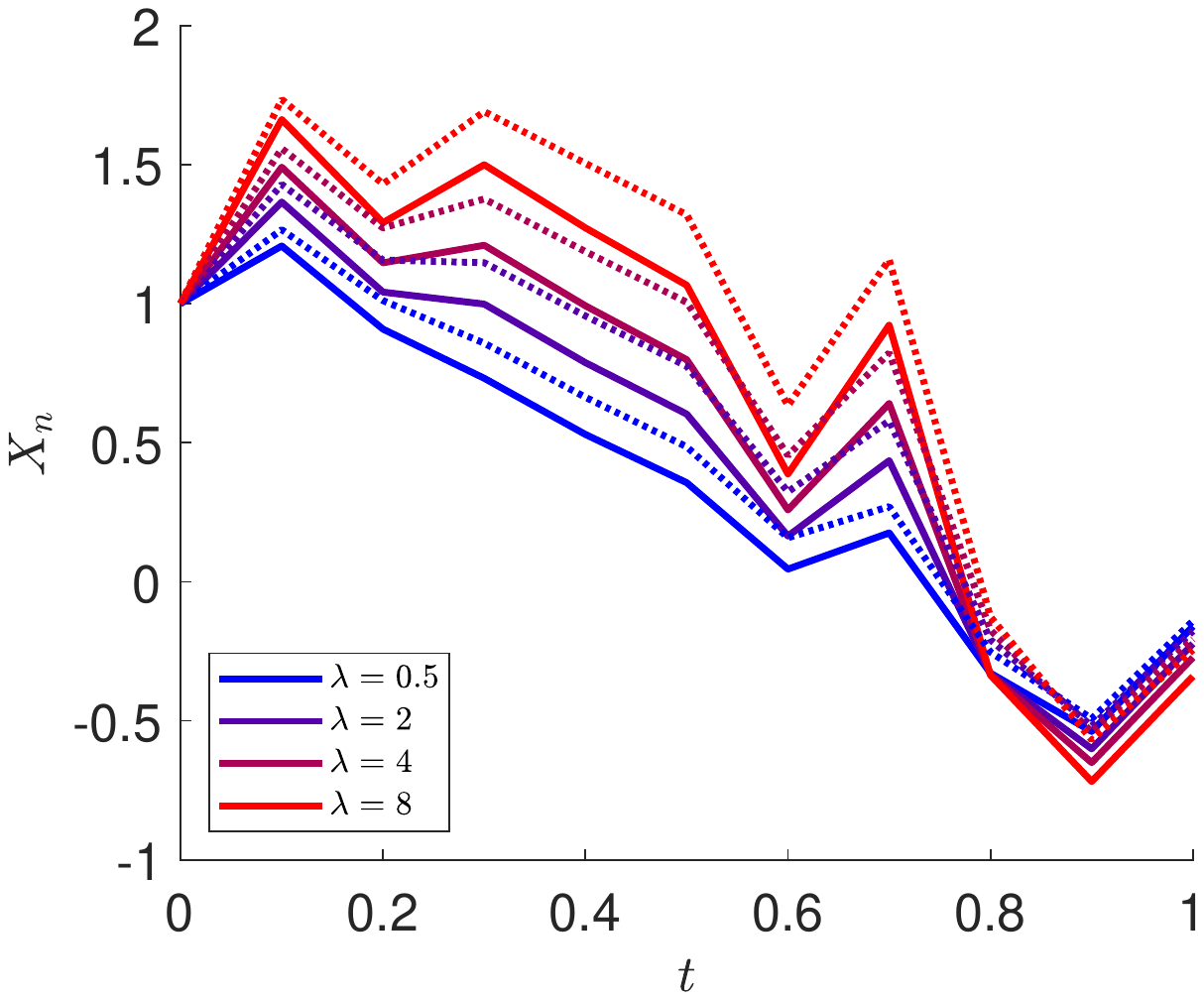}\hspace{10mm}
		\includegraphics[trim=140 240 140 240, scale=0.50]{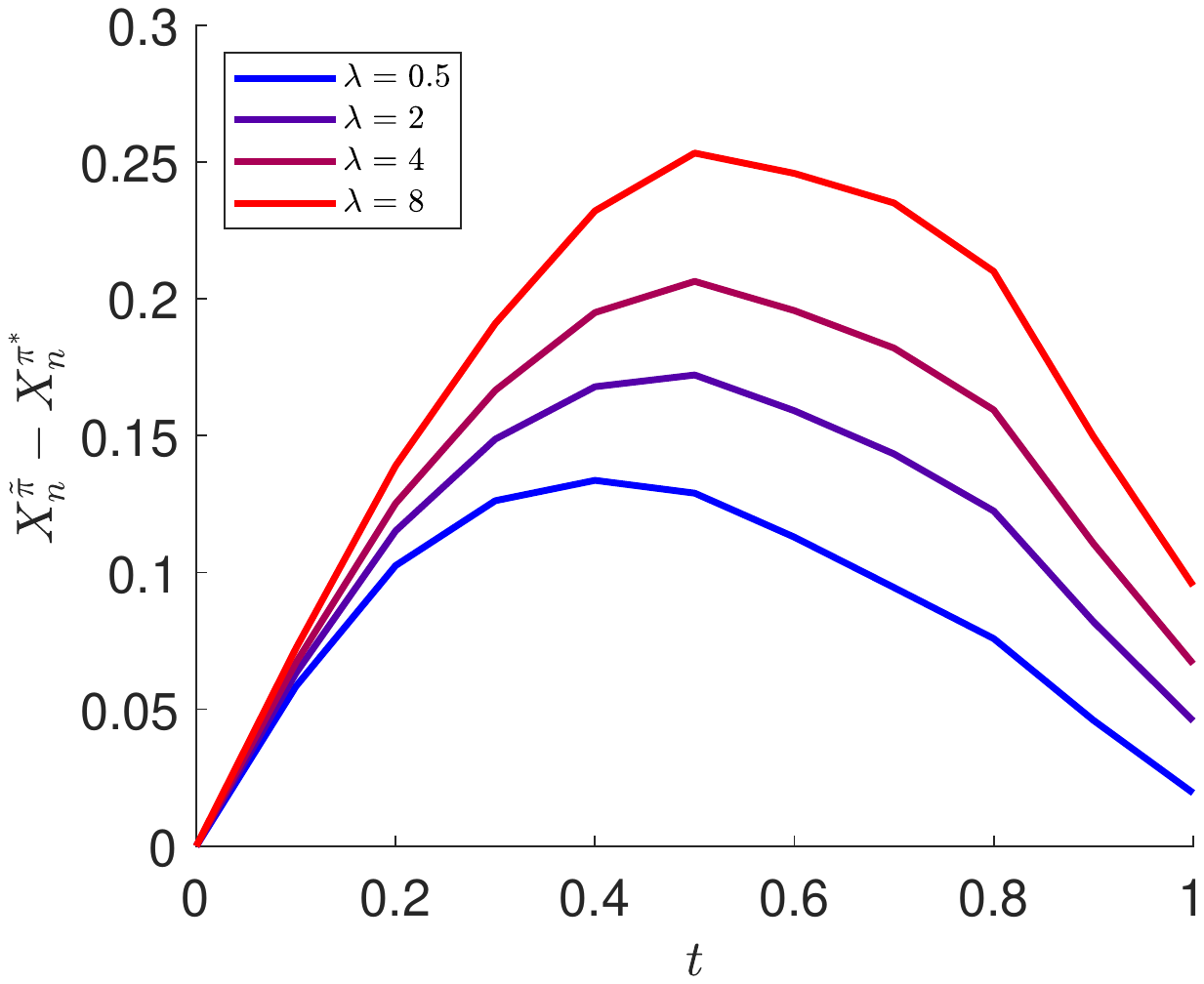}\\ \vspace{10mm}
		\includegraphics[trim=140 240 140 240, scale=0.50]{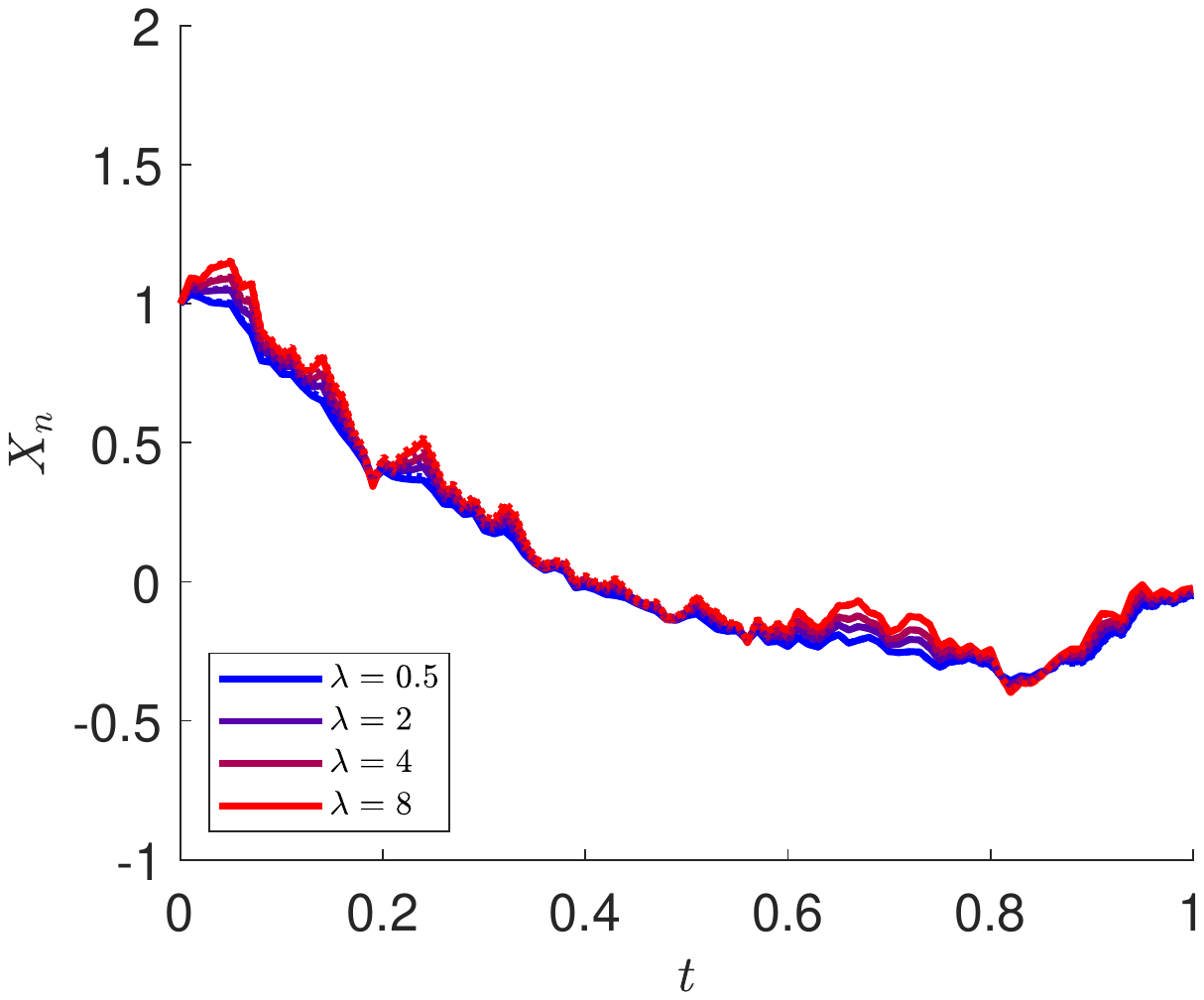}\hspace{10mm}
		\includegraphics[trim=140 240 140 240, scale=0.50]{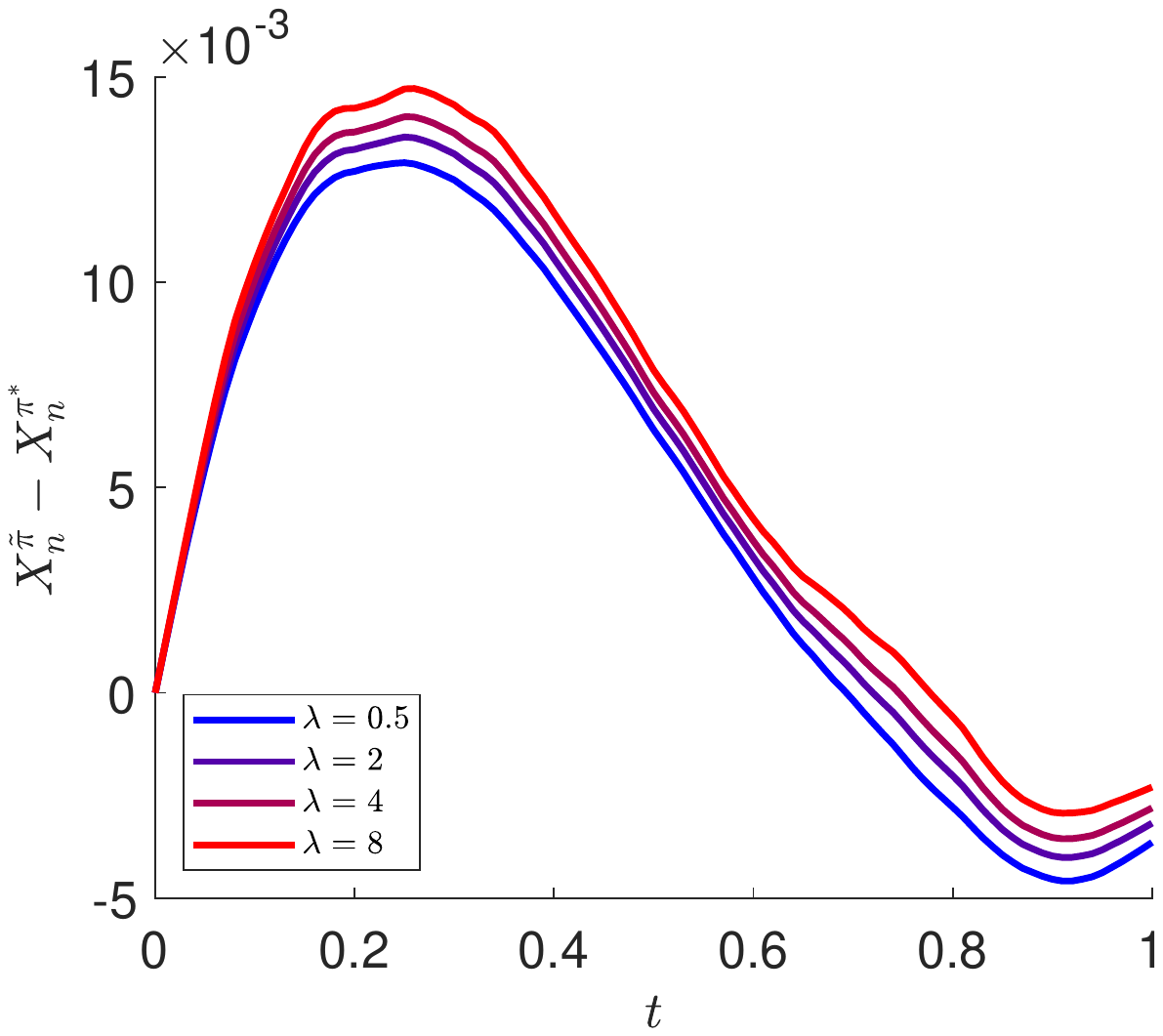}
	\end{center}
	\vspace{-1em}
	\caption{Approximation of discrete-time solution using closed form expressions relevant to continuous-time solution. In the top row, discretization is given by $T/\Delta t = 10$, and in the bottom row it is $T/\Delta t = 100$. The left column shows a realization of both the optimal discrete-time strategy $X_n^{\pi^*}$ (solid curves) and the approximation $X_n^{\tilde{\pi}}$ (dotted curves) resulting from implementing a strategy with $h^{(i)}_n$ replaced by $h^{(i)}_{t_n}$ in Proposition \ref{prop:discrete_optimal_exploration}. The right column shows the difference between these processes. Other parameter values are $B = 1$, $C=1$, $D=1$, $K=0.1$, $\gamma = 1$, $\sigma = 0.2$, $\kappa = 1$, $\eta = 2$, $\widehat{A}_0 = 0$, $\Sigma_0 = 1$, $q = 1.2$. \label{fig:discrete_approximation}}
\end{figure}

\section{Extension to $Q$-Learning}\label{sec:Q_learning}

In this section, we adapt the Tsallis entropy exploration reward to a model-agnostic reinforcement learning setting -- the so-called $Q$-learning paradigm. As such, this section adopts notation commonly used with Markov decision processes (MDPs) and $Q$-learning. For an overview see, e.g., \cite{jaimungal2022reinforcement}.

In a similar fashion to Sections \ref{sec:discrete_model_explore} and \ref{sec:continuous_model_explore}, we add a Tsallis entropy term to the classical agent's performance criterion  to incentivize exploration within the $Q$-learning framework. In this section, we work on discrete state and actions space denoted $\mcX$ and $\mcA$, respectively. Moreover, the agent chooses the distribution over actions conditional on state, denoted $\pi(\nu|X_n)$  -- here, as in soft $Q$-learning (See e.g., \cite{ziebart2010modeling, fox2015taming}), we call $\pi$ the policy.
The agent's Tsallis regularized performance criterion is defined as 
\begin{align*}
	V^\pi(x) &= \E^{\P_\pi}\left[
	\left.
	\sum_{n=1}^\infty \zeta^{n-1} \,(r_n + \lambda\,S_q[\pi_n(\cdot|X_n)])
	\,\right|\,X_1 = x\,\right],
\end{align*}
where $\lambda>0$ is the exploration reward parameter, $\zeta\in(0,1)$ is a discount factor, $S_q$ is the discrete analogue of the Tsallis entropy in Equation \eqref{eqn:entropy} given by
\begin{align}\label{eqn:entropy_discrete}
	S_q\left[\pi(\cdot|x)\right] &= 
	\left\{
	\begin{array}{cc}
		\displaystyle \frac{1}{q-1} \biggl(1 - \sum_{\nu\in\mcA} \left(\pi(\nu|x)\right)^q   \biggr)\,, 
		& q>0,  \,q\neq 1
		\\[2em]
		\displaystyle -\sum_{\nu\in\mcA} \pi_n(\nu|x)\,\log(\pi(\nu|x))\,, 
		& q = 1,
	\end{array}
	\right.
\end{align}
and the probability measure $\P_\pi$ is the one induced by the state transition probabilities $\P_\pi(X_{n+1}=x'\,|\,X=x)=\sum_{a\in\mcA} \P(X_{n+1}=x'\,|\,\nu=a,X=x)\pi(a\,|\,x)$. Defining the Tsallis regularized value function $V^*(x):=\sup_{\pi(\cdot|x)\in\mcP(\mcA)} V^\pi(x)$, where $\mcP(\mcA)$ denotes the space of probability measure on $\mcA$, 
the Tsallis regularized value function also satisfies a DPP which now takes the form
\begin{align*}
	V^*(x) &= \max_{\pi}\E^{\P_\pi}\biggl[r(x, X', \nu) + \zeta\,V^*(X') + \lambda\,S_q[\pi(\cdot|x)] \biggr]\,,
\end{align*}
where $(X',\nu)$ has distribution $\P_\pi(X'=x',\nu=a)= \P(X_{n+1}=x'|\nu=a,X=x)\pi(a\,|\,x)$.
Next, define the state-policy free energy function, also called the soft $Q$-function or $G$ function, as the value of taking an arbitrary policy (distribution over actions) $\pi\in\mcP$, for the current action, and then subsequently follow the, as yet unknown, optimal policy $\pi^*$. Then, the DPP implies that 
\begin{align}
	Q(x;\pi) :=&\, \E^{\P_\pi}\biggl[r(x,X',\nu) + \zeta\,V^*(X') + \lambda\,S_q[\pi(\cdot|x)]\biggr]\label{eqn:Q_explore}
	\\
	=&\, \E^{\P_\pi}\biggl[r(x,X',\nu) + \zeta\,\max_{\pi'\in\mcP}Q^*(X';\pi') + \lambda\,S_q[\pi(\cdot|x)]\biggr]\,,\nonumber
\end{align}
As the $Q$-function is defined for an arbitrary $\pi\in\mcP$, it may be evaluated on a point mass $\pi=\delta_a$ at a particular action $a\in\mcA$. Evaluating the $Q$-function with this particular choice leads to the classical (no exploration) $Q$-function and such a choice will be useful for the representation of the optimal exploratory strategy below.

For $q\neq 1$, expanding the expectation in $\eqref{eqn:Q_explore}$ yields
\begin{align}
	Q(x;\pi) &= \sum_{x'\in\mcX}\sum_{\nu\in\mcA} \biggl( r(x,x',\nu) + \zeta\,V^*(x') \biggr)\,\P(x'|x,\nu)\,\pi(\nu|x) + \frac{\lambda}{q-1}\biggl(1 - \sum_{\nu\in\mcA} (\pi(\nu|x))^q\biggr)\,.\label{eqn:Q_expanded}
\end{align}
Evaluating \eqref{eqn:Q_expanded} at $\pi = \delta_{a}$ for a particular $a \in\mcA$ gives
\begin{align}
	Q(x;\delta_{a}) &= \sum_{x'} \biggl( r(x,x',a) + \zeta\,V^*(x') \biggr)\,\P(x'|x,a)\,,
\end{align}
thus \eqref{eqn:Q_expanded} may be rewritten as
\begin{align}
	Q(x;\pi) &= \sum_{\nu\in\mcA} Q(x;\delta_\nu)\,\pi(\nu|x) + \frac{\lambda}{q-1}\biggl(1 - \sum_{\nu\in\mcA} (\pi(\nu|x))^q\biggr)\,.
\end{align}
Maximizing $Q(x;\pi)$ over $\pi$ gives the optimal distribution as\footnote{See Appendix \ref{proof:pi_star_Q} for these computations.}
\begin{align}\label{eqn:pi_star_Q}
	\pi^*(\nu|x) &= \left\{\begin{array}{cc}
			\biggl(\frac{q-1}{\lambda\,q}\biggr)^{\frac{1}{q-1}}\,\biggl[\psi(x) + Q(x;\delta_\nu)\biggr]^{\frac{1}{q-1}}_+\,, & q>1,
			\\[1em]
			\psi(x)\,\exp\biggl\{\frac{Q(x;\delta_\nu)}{\lambda}\biggr\}\,,
			& q = 1,
			\\[1em]
			\biggl(\frac{1-q}{\lambda\,q}\biggr)^{\frac{1}{q-1}}\,\biggl[\psi(x) - Q(x;\delta_\nu)\biggr]^{\frac{1}{q-1}}\,, & 0<q<1,
	\end{array}\right.
\end{align}
where $\psi(x)$ is a constant (for each $x\in\mcX$) chosen so that $\sum_{\nu\in\mcA} \pi^*(\nu|x) = 1$.
The case $q=1$ reduces to the so-called ``Boltzman'' policies obtained in \cite{fox2015taming}.

From the form of $\pi^*$ in \eqref{eqn:pi_star_Q} we see that the mode of the optimal policy is the value of $\nu$ which maximizes $Q(x;\delta_\nu)$, similar to the result in Propositions \ref{prop:discrete_optimal_exploration} and \ref{prop:continuous_optimal_exploration} where the location parameter of $\pi^*$ has the same feedback form as the optimal control $\nu^*$ of the standard problem without exploration. Indeed, the expression for $q=1$ is the softmax function applied to $Q(x;\delta_\nu)$, whereas for $q\neq1$ the expressions correspond to $q$-exponential variations of the softmax function. Further, for $q>1$, some actions may be assigned zero probability by $\pi^*$, similar to the behaviour seen with $q>1$ in equations \eqref{eqn:pi_star} and \eqref{eqn:pi_star_cont} where the optimal distribution has compact support. Unlike the result in those equations, however, the exclusion of some actions is not guaranteed. If the values of $Q(x;\delta_\nu)$ do not vary significantly with respect to $\nu$ or if $\lambda$ is relatively large, then the policy in \eqref{eqn:pi_star_Q} will assign positive probability to each action. When $q<1$ all actions are assigned strictly positive probability, similar to \eqref{eqn:pi_star} and \eqref{eqn:pi_star_cont} where the optimal policy has support on all of $\mathbb{R}$.

This optimal strategy may be used to develop a iterative $Q$-learning paradigm along the lines of \cite{todorov2005generalized} (see also the discussion on iLQG in \cite{jaimungal2022reinforcement}) but is outside the scope of this paper.

\section{Conclusions}\label{sec:conclusion}

We developed a grounded approach for exploratory controls in discrete-time with latent factors, and illuminated the need to carefully specify the information structure in specifying randomized controls. We also introduced Tsallis entropy as a reward for exploration so that the agent has the ability to further tailor the distribution of their randomized actions, generalizing beyond the Gaussian distribution to the class of $q$-Gaussian distributions. In particular, the agent may restrict their randomized actions to a compact set to cap the riskiness associated with a single action or may implement a distribution with fat tails to exhibit more exploration, depending on the specific problem setting. A comparison of the optimal exploratory control in discrete-time with that of an analogous model in continuous-time shows that the more tractable continuous-time solution may be used as an approximation to the more realistic discrete-time formulation. Finally, we show how these techniques may be adapted in the context of $Q$-learning so that the implemented actions employ more reasonable exploration rather than being maximally greedy.

Future work includes using our framework in a continual learning setting, where an agent first makes an approximation of environment that falls within our assumptions, then explores using our optimal control, and then updates the model with the new information they have received from the environment, re-optimise using our optimal exploratory control, and repeat. Further directions left open include how to generalize these results to the multi-agent case and/or mean-field game settings. As well, as applying the approach in real world settings on  data.

%\appendix
\section*{Appendix}
\renewcommand{\thesubsection}{\Alph{subsection}}

\subsection{Proof of Proposition \ref{prop:discrete_optimal_exploration}}
\label{proof:exploratory_discrete}
 
The dynamic programming principle for $H^{(\lambda)}$ in \eqref{eqn:exploration_dynamic_value} takes the form
\begin{align}
	H^{(\lambda)}_n &= \sup_{\pi_n\in\mathcal{A}^r_n}\E\biggl[H^{(\lambda)}_{n+1} + (D_n\,X_n\,\nu_n - C_n\,X_n^2 - K_n\,\nu_n^2)\,\Delta t + \lambda\,S_q[\pi_n]\,\Delta t\biggl|\mcF^u_n\biggr]\,,\label{eqn:DPP}
\end{align}
where we suppress dependence of each process on $\pi_n$ to simplify notation. We make the ansatz $H^{(\lambda)}_n = h^{(0)}_n + h^{(1)}_n\,X_n + h^{(2)}_n\,X_n^2$ where $h^{(0)}$ and $h^{(1)}$ are $\mcF_n$-adapted (the filtration generated by $Y$) and $h^{(2)}$ is deterministic. Substituting this ansatz for $H^{(\lambda)}$ on both sides of \eqref{eqn:DPP} and using equations \eqref{eqn:Y_filter} and \eqref{eqn:relaxed_X} yields
\begin{equation}
	\begin{split}
		& h^{(0)}_n + h^{(1)}_n\,X_n + h^{(2)}_n\,X_n^2\\
		& \hspace{10mm} = \sup_{\pi_n\in\mathcal{A}^r_n} \biggl\{ \E[h^{(0)}_{n+1}|\mcF^u_n] + \E[h^{(1)}_{n+1}|\mcF^u_n]\,X_n + \E[h^{(1)}_{n+1}|\mcF^u_n]\,\widehat{A}_n\,\Delta t + \E[h^{(1)}_{n+1}\,\nu_n|\mcF^u_n]\,\gamma_n\,\Delta t \\
		& \hspace{15mm} + \E[h^{(1)}_{n+1}\,\Delta \widehat{M}_n|\mcF^u_n] + h^{(2)}_{n+1}\, \biggl(X_n^2 + \widehat{A}_n^2\,(\Delta t)^2 + \E[\nu_n^2|\mcF^u_n]\,\gamma_n^2\,(\Delta t)^2 + \widehat{\sigma}_n^2\,\Delta t\\
		& \hspace{15mm} + 2\,X_n\,\widehat{A}_n\,\Delta t + 2\,X_n\,\gamma_n\,\E[\nu_n|\mcF^u_n]\,\Delta t + 2\,X_n\,\E[\Delta\widehat{M}_n|\mcF^u_n] + 2\,\widehat{A}_n\,\gamma_n\,\E[\nu_n|\mcF^u_n]\,(\Delta t)^2 \\
		& \hspace{15mm} + 2\,\widehat{A}_n\E[\Delta\widehat{M}_n|\mcF^u_n]\,\Delta t + \gamma_n\,\E[\nu_n\,\Delta\widehat{M}_n|\mcF^u_n]\,\Delta t  \biggr) + D_n\,X_n\,\E[\nu_n|\mcF^u_n]\,\Delta t\\
		& \hspace{15mm} + C_n\,X_n^2\,\Delta t - K_n\,\E[\nu_n^2|\mcF^u_n]\,\Delta t + \lambda\,S_q[\pi_n]\,\Delta t  \biggr\}\,.
	\end{split}\label{eqn:DPP_expand}
\end{equation}
where $\Delta\widehat{M}_n = \widehat{M}_{n+1} - \widehat{M}_n$ and $\widehat{\sigma}_n^2 = \E[(\Delta\widehat{M}_n)^2|\mcF^u_n]$. Part of the ansatz that has been made is that $h^{(0)}$ and $h^{(1)}$ are $\mcF_n$-adapted, hence
\begin{align*}
	\E[h^{(0)}_{n+1}|\mcF^u_n] &= \E[h^{(0)}_{n+1}|\mcF_n]\\
	\E[h^{(1)}_{n+1}|\mcF^u_n] &= \E[h^{(1)}_{n+1}|\mcF_n]\\
	\E[h^{(1)}_{n+1}\,\Delta\widehat{M}_n|\mcF^u_n] &= \E[h^{(1)}_{n+1}\,\Delta\widehat{M}_n|\mcF_n]
\end{align*} 
by the independence of $Y$ and $\{u_n\}_{n=0}^{N-1}$. We also have that $\E[\Delta\widehat{M}_n|\mcF^u_n] = 0$ in \eqref{eqn:DPP_expand}. Other relevant conditional expectations in \eqref{eqn:DPP_expand} can be computed as
\begin{align*}
	\E[h^{(1)}_{n+1}\,\nu_n|\mcF^u_n] &= \E\biggl[\E[h^{(1)}_{n+1}\,\nu_n|\mcF^u_n, u_n]\biggl|\mcF^u_n\biggr]\\
	&= \E\biggl[\nu_n\,\E[h^{(1)}_{n+1}|\mcF^u_n,u_n]\biggl|\mcF^u_n\biggr]\\
	&= \E[\nu_n|\mcF^u_n]\,\E[h^{(1)}_{n+1}|\mcF^u_n,u_n]\\
	&= \int_{-\infty}^\infty \nu\,\pi_n(\nu)\,d\nu\,\E[h^{(1)}_{n+1}|\mcF_n]
 \end{align*}
 and
 \begin{align*}
	\E[\nu_n\,\Delta\widehat{M}_n|\mcF^u_n] &= 	\E\biggl[\E[\nu_n\,\Delta\widehat{M}_n|\mcF^u_n,u_n]\biggl|\mcF^u_n\biggr]\\
	&= \E\biggl[\nu_n\E[\Delta\widehat{M}_n|\mcF^u_n,u_n]\biggl|\mcF^u_n\biggr]\\
	&= 0\,.
\end{align*}
By using \eqref{eqn:nu_bar} and \eqref{eqn:nu_sigma}, and grouping terms by the power of $\nu$, we rewrite \eqref{eqn:DPP_expand} as
\begin{equation}
	\begin{split}
		& h^{(0)}_n + h^{(1)}_n\,X_n + h^{(2)}_n\,X_n^2\\
		& \hspace{10mm} = \E[h^{(0)}_{n+1}|\mcF_n] + \E[h^{(1)}_{n+1}|\mcF_n]\,(X_n + \widehat{A}_n\,\Delta t) + \E[h^{(1)}_{n+1}\,\Delta\widehat{M}_n|\mcF_n]\\
		& \hspace{15mm} + h^{(2)}_{n+1}\,\biggl(X_n^2 + \widehat{A}_n^2\,(\Delta t)^2 + \widehat{\sigma}_n^2\,\Delta t + 2\,X_n\,\widehat{A}_n\,\Delta t \biggr) - C_n\,X_n^2\,\Delta t\\
		& \hspace{15mm} + \sup_{\pi_n\in\mathcal{A}^r_n}\biggl\{ \biggl( \E[h^{(1)}_{n+1}|\mcF_n]\,\gamma_n + 2\,h^{(2)}_{n+1}\,\widehat{A}_n\,\gamma_n\,\Delta t + (2\,h^{(2)}_{n+1}\,\gamma_n + D_n)\,X_n     \biggr)\,\int_{-\infty}^\infty \nu\,\pi_n(\nu)\,d\nu \\
		& \hspace{15mm} - (K_n - h^{(2)}_{n+1}\,\gamma_n^2\,\Delta t)\,\int_{-\infty}^\infty \nu^2\,\pi_n(\nu)\,d\nu + \lambda\,S_q[\pi_n] \biggr\}\,\Delta t\,.
	\end{split}
\end{equation}
We introduce $L_n = \E[h^{(1)}_{n+1}|\mcF_n]\,\gamma_n + 2\,h^{(2)}_{n+1}\,\widehat{A}_n\,\gamma_n\,\Delta t + (2\,h^{(2)}_{n+1}\,\gamma_n + D_n)\,X_n$ and $Q_n = K_n - h^{(2)}_{n+1}\,\gamma_n^2\,\Delta t$ then complete the square with respect to $\nu$ to write
\begin{equation}
	\begin{split}
		& h^{(0)}_n + h^{(1)}_n\,X_n + h^{(2)}_n\,X_n^2\\
		& \hspace{10mm} = \E[h^{(0)}_{n+1}|\mcF_n] + \E[h^{(1)}_{n+1}|\mcF_n]\,(X_n + \widehat{A}_n\,\Delta t) + \E[h^{(1)}_{n+1}\,\Delta\widehat{M}_n|\mcF_n]\\
		& \hspace{15mm} + h^{(2)}_{n+1}\,\biggl(X_n^2 + \widehat{A}_n^2\,(\Delta t)^2 + \widehat{\sigma}_n^2\,\Delta t + 2\,X_n\,\widehat{A}_n\,\Delta t \biggr) - C_n\,X_n^2\,\Delta t\\
		& \hspace{15mm} + \frac{L_n^2}{4\,Q_n}\,\Delta t + \sup_{\pi_n\in\mathcal{A}^r_n}\biggl\{\lambda\,S_q[\pi_n] - \int_{-\infty}^\infty Q_n\,\biggl(\nu - \frac{L_n}{2\,Q_n}\biggr)^2\,\pi_n(\nu)\,d\nu\biggr\}\,\Delta t
	\end{split}\label{eqn:DPP_expand2}
\end{equation}
To perform the maximization in \eqref{eqn:DPP_expand2}, we introduce a Lagrange multiplier $\psi_n$ to enforce the constraint $\int_{-\infty}^\infty \pi_n(\nu)\,d\nu = 1$ and for each $\nu$ a KKT multiplier $\zeta_n(\nu)$ for the constraint $\pi_n(\nu) \geq 0$. The Lagrangian associated with this (strictly convex) constrained maximization problem is
\begin{equation}
	\begin{split}
		\mathcal{L}(\pi_n;\psi_n,\zeta_n) &= \int_{-\infty}^\infty -Q_n\,\biggl(\nu - \frac{L_n}{2\,Q_n}\biggr)^2\,\pi_n(\nu)\,d\nu + \frac{\lambda}{q-1}\,\biggl(1 - \int_{-\infty}^\infty \pi^q_n(\nu)\,d\nu\biggr)\\
		&\hspace{10mm} + \psi_n\biggl(\int_{-\infty}^\infty \pi_n(\nu)\,d\nu - 1\biggr) + \int_{-\infty}^\infty \zeta_n(\nu)\,\pi_n(\nu)\,d\nu\,.
	\end{split}
\end{equation}
The Gateaux derivative of $\mathcal{L}$ with respect to $\pi_n$ is given by
\begin{equation}
	\begin{split}
		<D\mathcal{L}(\pi_n;\psi_n,\zeta_n), f> &= \int_{-\infty}^\infty f(\nu)\,\biggl( \zeta_n(\nu) + \psi_n - \frac{\lambda\,q}{q-1}\,\pi^{q-1}_n(\nu) - Q_n\,(\nu - \frac{L_n}{2\,Q_n})^2     \biggr)\,d\nu\,.
	\end{split}
\end{equation}
In order to attain a maximizer, first order conditions imply that the conditional density $\pi_n$, Lagrange multiplier $\psi_n$, and KKT multiplier $\zeta_n$ must then satisfy
\begin{align}
	\pi^*_n(\nu) &= \biggl(\frac{q-1}{\lambda\,q}\biggr)^{\frac{1}{q-1}}\biggl(\zeta_n(\nu) + \psi_n - Q_n\,\biggl(\nu - \frac{L_n}{2\,Q_n}\biggr)^2\biggr)^{\frac{1}{q-1}}\,.\label{eqn:first_order_condition}
\end{align}
Additionally, the KKT multiplier must be chosen so that $\zeta_n(\nu)\geq 0$, $\pi_n(\nu)\geq 0$, and $\zeta_n(\nu)\,\pi_n(\nu) = 0$. This is seen to be achieved by
\begin{align}
	\zeta_n(\nu) &= \biggl[Q_n\biggl(\nu - \frac{L_n}{2\,Q_n}\biggr)^2 - \psi_n\biggr]_+\,,
\end{align}
which reduces \eqref{eqn:first_order_condition} to
\begin{align}
	\pi^*_n(\nu) &= \biggl(\frac{q-1}{\lambda\,q}\biggr)^{\frac{1}{q-1}}\biggl[\psi_n - Q_n\,\biggl(\nu - \frac{L_n}{2\,Q_n}\biggr)^2\biggr]^{\frac{1}{q-1}}_+\,.\label{eqn:pi_proof}
\end{align}
The Lagrange multiplier $\psi_n$ is found by enforcing $\int_{-\infty}^\infty \pi_n(\nu)\,d\nu = 1$. Using a simple substitution, this integral can be written as
\begin{align*}
	\int_{-\infty}^\infty \pi_n(\nu)\,d\nu &= \biggl(\frac{q-1}{\lambda\,q}\biggr)^{\frac{1}{q-1}}\,\frac{\psi_n^{\frac{1}{q-1}+\frac{1}{2}}}{\sqrt{Q_n}}\,\int_{-1}^1 (1-x^2)^{\frac{1}{q-1}}\,dx\\
	1 &= \biggl(\frac{q-1}{\lambda\,q}\biggr)^{\frac{1}{q-1}}\,\frac{\psi_n^{\frac{1}{q-1}+\frac{1}{2}}}{\sqrt{Q_n}}\,\frac{\sqrt{\pi}\,\Gamma(\frac{1}{q-1}+1)}{\Gamma(\frac{1}{q-1} + \frac{3}{2})}\,,
\end{align*}
where the integral with respect to $x$ can be found in \cite{gradshteyn2014table} Section 3.251. Solving for $\psi_n$ gives
\begin{align*}
	\psi_n &= \Biggl[\frac{\Gamma(\frac{1}{q-1}+\frac{3}{2})}{\sqrt{\pi}\,\Gamma(\frac{1}{q-1}+1)}\biggl(\frac{\lambda\,q}{q-1}\biggr)^{\frac{1}{q-1}}\sqrt{K_n - h^{(2)}_{n+1}\,\gamma_n^2\,\Delta t}\Biggr]^{\frac{1}{\frac{1}{q-1}+\frac{1}{2}}}\,,
\end{align*}
as in the statement of the proposition. Substituting the maximizer $\pi^*_n$ from \eqref{eqn:pi_proof} into the supremum term in \eqref{eqn:DPP_expand2} and evaluating the integrals using \cite{gradshteyn2014table} Section 3.251 gives
\begin{equation}
	\begin{split}
		& \sup_{\pi_n\in\mathcal{A}^r_n} \biggl\{\lambda\,S_q[\pi_n] - \int_{-\infty}^\infty Q_n\,\biggl(\nu - \frac{L_n}{2\,Q_n}\biggr)^2\,\pi_n(\nu)\,d\nu\biggr\}\\
		& \hspace{15mm} = \frac{\lambda}{q-1} - \biggl(\frac{q-1}{\lambda\,q}\biggr)^{\frac{1}{q-1}}\frac{\psi_n^{\frac{1}{q-1}+\frac{3}{2}}\,\sqrt{\pi}}{\sqrt{Q_n}\,\Gamma(\frac{1}{q-1}+\frac{5}{2})}\biggl(\frac{1}{2}\,\Gamma\biggl(\frac{1}{q-1}+1\biggr) + \frac{1}{q}\,\Gamma\biggl(\frac{1}{q-1} + 2\biggr)\biggr)\\
		& \hspace{15mm} = \frac{\lambda}{q-1} - \biggl(\frac{q-1}{\lambda\,q}\biggr)^{\frac{1}{q-1}}\frac{\psi_n^{\frac{1}{q-1}+\frac{3}{2}}\,\sqrt{\pi}}{\sqrt{Q_n}\,\Gamma(\frac{1}{q-1}+\frac{5}{2})}\biggl(\frac{1}{2}\,\Gamma\biggl(\frac{1}{q-1}+1\biggr) + \frac{1}{q-1}\,\Gamma\biggl(\frac{1}{q-1} + 1\biggr)\biggr)\\
		& \hspace{15mm} = \frac{\lambda}{q-1} - \biggl(p+\frac{1}{2}\biggr)\,\biggl(\frac{q-1}{\lambda\,q}\biggr)^p \frac{\psi_n^{p+\frac{3}{2}}}{\sqrt{Q_n}}\,\frac{\sqrt{\pi}\,\Gamma(p+1)}{\Gamma(p+\frac{5}{2})}\,,
	\end{split}\nonumber
\end{equation}
where $p = \frac{1}{q-1}$. Thus the dynamic programming principle now reads
\begin{equation}
	\begin{split}
		& h^{(0)}_n + h^{(1)}_n\,X_n + h^{(2)}_n\,X_n^2\\
		& \hspace{10mm} = \E[h^{(0)}_{n+1}|\mcF_n] + \E[h^{(1)}_{n+1}|\mcF_n]\,(X_n + \widehat{A}_n\,\Delta t) + \E[h^{(1)}_{n+1}\,\Delta\widehat{M}_n|\mcF_n]\\
		& \hspace{15mm} + h^{(2)}_{n+1}\,\biggl(X_n^2 + \widehat{A}_n^2\,(\Delta t)^2 + \widehat{\sigma}_n^2\,\Delta t + 2\,X_n\,\widehat{A}_n\,\Delta t \biggr) - C_n\,X_n^2\,\Delta t\\
		& \hspace{15mm} + \frac{(\E[h^{(1)}_{n+1}|\mcF_n]\,\gamma_n + 2\,h^{(2)}_{n+1}\,\widehat{A}_n\,\gamma_n\,\Delta t + (2\,h^{(2)}_{n+1}\,\gamma_n + D_n)\,X_n)^2}{4\,(K_n - h^{(2)}_{n+1}\,\gamma_n^2\,\Delta t)}\,\Delta t\\
		& \hspace{15mm} + \frac{\lambda}{q-1}\,\Delta t - \biggl(p+\frac{1}{2}\biggr)\,\biggl(\frac{q-1}{\lambda\,q}\biggr)^p \frac{\psi_n^{p+\frac{3}{2}}}{\sqrt{K_n - h^{(2)}_{n+1}\,\gamma_n^2\,\Delta t}}\,\frac{\sqrt{\pi}\,\Gamma(p+1)}{\Gamma(p+\frac{5}{2})}\,\Delta t\,.
	\end{split}
\end{equation}
Expanding and grouping by powers of $X_n$ gives
\begin{equation}
	\begin{split}
		& h^{(0)}_n + h^{(1)}_n\,X_n + h^{(2)}_n\,X_n^2\\
		& \hspace{10mm} = \E[h_{n+1}^{(0)}|\mcF_n] + \E[h_{n+1}^{(1)}|\mcF_n]\,\widehat{A}_n\,\Delta t + \E[h_{n+1}^{(1)}\,\Delta \widehat{M}_n|\mcF_n] + h_{n+1}^{(2)}\,\widehat{A}_n^2\,(\Delta t)^2\\
		& \hspace{15mm} + h_{n+1}^{(2)}\,\widehat{\sigma}_n^2\,\Delta t + \frac{(\E[h_{n+1}^{(1)}|\mcF_n]\,\gamma_n + 2\,h_{n+1}^{(2)}\,\widehat{A}_n\,\gamma_n\,\Delta t)^2}{4\,(K_n - h_{n+1}^{(2)}\,\gamma_n^2\,\Delta t)}\,\Delta t + \frac{\lambda}{q-1}\,\Delta t\\
		& \hspace{15mm} - \biggl(p+\frac{1}{2}\biggr)\,\biggl(\frac{q-1}{\lambda\,q}\biggr)^p \frac{\psi_n^{p+\frac{3}{2}}}{\sqrt{K_n - h^{(2)}_{n+1}\,\gamma_n^2\,\Delta t}}\,\frac{\sqrt{\pi}\,\Gamma(p+1)}{\Gamma(p+\frac{5}{2})}\,\Delta t\\
		& \hspace{15mm} + \biggl(\E[h_{n+1}^{(1)}|\mcF_n] + 2\,h_{n+1}^{(2)}\,\widehat{A}_n\,\Delta t + \frac{(\E[h_{n+1}^{(1)}|\mcF_n]\,\gamma_n + 2\,h_{n+1}^{(2)}\,\widehat{A}_n\,\Delta t)\,(2\,h_{n+1}^{(2)}\,\gamma_n + D_n)}{2\,(K_n - h_{n+1}^{(2)}\,\gamma_n\,\Delta t)}\,\Delta t\biggr)\,X_n\\
		& \hspace{15mm} + \biggl(h_{n+1}^{(2)} - C_n\,\Delta t + \frac{(2\,h_{n+1}^{(2)}\,\gamma_n + D_n)^2}{4\,(K_n - h_{n+1}^{(2)}\,\gamma_n\,\Delta t)}\,\Delta t\biggr)\,X_n^2\,.
	\end{split}
\end{equation}
By matching the coefficients on each power of $X_n$ we arrive at the recursion equations
\begin{subequations}
	\begin{align}
		\begin{split}
			h_n^{(0)} &= \E[h_{n+1}^{(0)}|\mcF_n] + \E[h_{n+1}^{(1)}|\mcF_n]\,\widehat{A}_n\,\Delta t + \E[h_{n+1}^{(1)}\,\Delta \widehat{M}_n|\mcF_n] + h_{n+1}^{(2)}\,\widehat{A}_n^2\,(\Delta t)^2\\
				& \hspace{10mm} + h_{n+1}^{(2)}\,\widehat{\sigma}_n^2\,\Delta t + \frac{(\E[h_{n+1}^{(1)}|\mcF_n]\,\gamma_n + 2\,h_{n+1}^{(2)}\,\widehat{A}_n\,\gamma_n\,\Delta t)^2}{4\,(K_n - h_{n+1}^{(2)}\,\gamma_n^2\,\Delta t)}\,\Delta t + \frac{\lambda}{q-1}\,\Delta t\\
				& \hspace{10mm} - \biggl(p+\frac{1}{2}\biggr)\,\biggl(\frac{q-1}{\lambda\,q}\biggr)^p \frac{\psi_n^{p+\frac{3}{2}}}{\sqrt{K_n - h^{(2)}_{n+1}\,\gamma_n^2\,\Delta t}}\,\frac{\sqrt{\pi}\,\Gamma(p+1)}{\Gamma(p+\frac{5}{2})}\,\Delta t
		\end{split}\\
		\begin{split}
			h_n^{(1)} &= \E[h_{n+1}^{(1)}|\mcF_n] + 2\,h_{n+1}^{(2)}\,\widehat{A}_n\,\Delta t\\
			& \hspace{15mm} + \frac{(\E[h_{n+1}^{(1)}|\mcF_n]\,\gamma_n + 2\,h_{n+1}^{(2)}\,\widehat{A}_n\,\Delta t)\,(2\,h_{n+1}^{(2)}\,\gamma_n + D_n)}{2\,(K_n - h_{n+1}^{(2)}\,\gamma_n\,\Delta t)}\,\Delta t\\
		\end{split}\label{eqn:prf_BSDeltaE1} \\
		h_n^{(2)} &= h_{n+1}^{(2)} - C_n\,\Delta t + \frac{(2\,h_{n+1}^{(2)}\,\gamma_n + D_n)^2}{4\,(K_n - h_{n+1}^{(2)}\,\gamma_n\,\Delta t)}\,\Delta t\label{eqn:prf_BSDeltaE2}
	\end{align}
\end{subequations}
with terminal conditions $h_N^{(0)} = h_N^{(1)} = 0$ and $h_N^{(2)} = -B$. The optimizer $\pi^*_n$ in \eqref{eqn:pi_proof} does not depend on the process $h^{(0)}$, so we require only the recursion equations \eqref{eqn:prf_BSDeltaE1} and \eqref{eqn:prf_BSDeltaE2} for $h^{(1)}$ and $h^{(2)}$, which are identical to those in equation \eqref{eqn:BSDeltaE} as in the statement of the proposition. \qed

\subsection{Proof of Proposition \ref{prop:continuous_optimal_exploration}}
\label{proof:exploratory_cont}

The HJB equation with terminal condition associated with the dynamic value function $H^{(\lambda)}$ is
\begin{equation}\label{eqn:HJB_pf}
	\begin{split}
		\partial_tH^{(\lambda)} + \mathcal{L}^{\widehat{A}}H^{(\lambda)} + \widehat{A}\,\partial_xH^{(\lambda)} + \frac{1}{2}\,\sigma^2\,\partial_{xx}H^{(\lambda)} + \sigma\,v\,\partial_{x\widehat{A}}H^{(\lambda)} - C_t\,x^2 \hspace{30mm} 
  \\
	 + \sup_{\pi_t\in\mathcal{A}^r_t}\biggl\{(\gamma_t\,\partial_xH^{(\lambda)} + D_t\,x)\int_{-\infty}^\infty \nu\,\pi_t(\nu)\,d\nu - K_t\int_{-\infty}^\infty\nu^2\,\pi_t(\nu)\,d\nu + \lambda\,S_q[\pi_t]\biggr\} &= 0\,, 
  \\
		 H^{(\lambda)}(T,x,\widehat{A}) &= g(x)\,.
	\end{split}
\end{equation}

This PDE can be rewritten in the more convenient form
\begin{equation}\label{eqn:HJB_pf2}
	\begin{split}
		\partial_tH^{(\lambda)} + \mathcal{L}^{\widehat{A}}H^{(\lambda)} + \widehat{A}\,\partial_xH^{(\lambda)} + \frac{1}{2}\,\sigma^2\,\partial_{xx}H^{(\lambda)} + \sigma\,v\,\partial_{x\widehat{A}}H^{(\lambda)} - C_t\,x^2 + \frac{\lambda}{q-1} \hspace{20mm}\\
		 +\, \frac{(\gamma_t\,\partial_xH^{(\lambda)} + D_t\,x)^2}{4\,K_t} + \sup_{\pi_t\in\mathcal{A}^r_t}\biggl\{\int_{-\infty}^\infty \biggl[-K_t \,\biggl(\nu - \frac{\gamma_t\,\partial_xH^{(\lambda)} + D_t\,x}{2\,K_t}\biggr)^2 - \frac{\lambda}{q-1}\,\pi_t^{q-1}(\nu)\biggr]\,\pi_t(\nu)\,d\nu \biggr\} &= 0\,.
	\end{split}
\end{equation}
Computing the supremum term above is similar to the computation performed in the proof of Proposition \ref{prop:discrete_optimal_exploration}. We repeat it here, but this time for the case $1/3<q<1$. To this end, introduce a Lagrange multiplier $\psi_t$ to enforce the constraint $\int_{-\infty}^\infty \pi_t(\nu)\,d\nu=1$ and for each $\nu$ introduce a KKT multiplier $\zeta_t(\nu)$ for the constraint $\pi_t(\nu)\geq 0$. The Lagrangian associated with this (strictly convex) constrained maximization is
\begin{equation}
    \begin{split}
        \mathcal{L}(\pi_t;\psi_t,\zeta_t) &= \int_{-\infty}^\infty \biggl[-K_t \,\biggl(\nu - \frac{L_t}{2\,K_t}\biggr)^2 - \frac{\lambda}{q-1}\,\pi_t^{q-1}(\nu)\biggr]\,d\nu\\
        & \hspace{10mm} - \psi_t\biggl(\int_{-\infty}^\infty \pi_t(\nu)\,d\nu - 1\biggr) + \int_{-\infty}^\infty \zeta_t(\nu)\,\pi_t(\nu)\,d\nu\,,
    \end{split}    
\end{equation}
where $L_t = \gamma_t\,\partial_xH^{(\lambda)} + D_t\,x$. The Gateaux derivative of $\mathcal{L}$ with respect to $\pi_t$ is given by
\begin{equation}
    \begin{split}
        <D\mathcal{L}(\pi_n;\psi_n,\zeta_n), f> &= \int_{-\infty}^\infty f(\nu)\,\biggl( \zeta_t(\nu) - \psi_t - \frac{\lambda\,q}{q-1}\pi_t^{q-1}(\nu) - K_t\biggl(\nu - \frac{L_t}{2\,K_t}\biggr)^2    \biggr)\,d\nu
    \end{split}    
\end{equation}
To attain a maximizer, first order conditions imply that the conditional density $\pi_t$, Lagrange multiplier $\psi_t$, and KKT multiplier $\zeta_t$ must satisfy
\begin{align}
    \pi_t^*(\nu) &= \biggl(\frac{1-q}{\lambda\,q}\biggr)^{\frac{1}{q-1}}\biggl(-\zeta_t(\nu) + \psi_t + K_t\biggl(\nu - \frac{L_t}{2\,K_t}\biggr)^2\biggr)^{\frac{1}{q-1}}\,.
\end{align}
Since $q<1$ we have $\pi_t^*(\nu) \neq 0$ and so we conclude $\zeta_t(\nu) = 0$, which reduces the density to
\begin{align}\label{eqn:cont_density_pf}
    \pi_t^*(\nu) &= \biggl(\frac{1-q}{\lambda\,q}\biggr)^{\frac{1}{q-1}}\biggl(\psi_t + K_t\biggl(\nu - \frac{L_t}{2\,K_t}\biggr)^2\biggr)^{\frac{1}{q-1}}\,.
\end{align}
Then $\psi_t$ is chosen to enforce the constraint $\int_{-\infty}^\infty \pi_t(\nu)\,d\nu$. A substitution allows this integral to be rewritten as
\begin{align*}
    \int_{-\infty}^\infty \pi_t(\nu)\,d\nu &= \biggl(\frac{1-q}{\lambda\,q}\biggr)^{\frac{1}{q-1}}\frac{\psi_t^{\frac{1}{q-1} + \frac{1}{2}}}{\sqrt{K_t}}\,\int_{-\infty}^\infty (1+x^2)^{\frac{1}{q-1}}\,dx\\
    &= \biggl(\frac{1-q}{\lambda\,q}\biggr)^{\frac{1}{q-1}}\frac{\psi_t^{\frac{1}{q-1} + \frac{1}{2}}}{\sqrt{K_t}}\,\frac{\sqrt{\pi}\,\Gamma(\frac{1}{1-q}-\frac{1}{2})}{\Gamma(\frac{1}{1-q})}\,,
\end{align*}
where the integral with respect to $x$ can be found in \cite{gradshteyn2014table} Section 3.251. Solving for $\psi_t$ gives
\begin{align*}
    \psi_t &= \biggl[\frac{\Gamma(\frac{1}{1-q})}{\sqrt{\pi}\,\Gamma(\frac{1}{1-q} - \frac{1}{2})}\,\biggl(\frac{\lambda\,q}{1-q}\biggr)^{\frac{1}{q-1}}\sqrt{K_t}\biggr]^{\frac{1}{\frac{1}{q-1} + \frac{1}{2}}}\,,
\end{align*}
as in the statement of the Proposition. Substituting the density from \eqref{eqn:cont_density_pf} back into supremum term in \eqref{eqn:HJB_pf2} and referring to \cite{gradshteyn2014table} Section 3.251 to compute the integrals gives
\begin{equation}
    \begin{split}
        &\sup_{\pi_t\in\mathcal{A}^r_t}\biggl\{\int_{-\infty}^\infty \biggl[-K_t \,\biggl(\nu - \frac{\gamma_t\,\partial_xH^{(\lambda)} + D_t\,x}{2\,K_t}\biggr)^2 - \frac{\lambda}{q-1}\,(\pi_t^*(\nu))^{q-1}\biggr]\,\pi^*_t(\nu)\,d\nu \biggr\}\\
        =& -\biggl(\frac{1-q}{\lambda\,q}\biggr)^{\frac{1}{q-1}}\psi_t^{\frac{1}{q-1}+\frac{3}{2}}\frac{\sqrt{\pi}\,\Gamma(\frac{1}{1-q}-\frac{3}{2})}{\sqrt{K_t}}\biggl(\frac{1}{2\,\Gamma(\frac{1}{1-q})} - \frac{1}{(1-q)\,\Gamma(\frac{1}{1-q})}\biggr)\\
        =& \frac{q+1}{3\,q-1}\,\psi_t\,.
    \end{split}    
\end{equation}
Thus the HJB equation with terminal condition becomes
\begin{equation}\label{eqn:HJB_pf3}
	\begin{split}
		\partial_tH^{(\lambda)} + \mathcal{L}^{\widehat{A}}H^{(\lambda)} + \widehat{A}\,\partial_xH^{(\lambda)} + \frac{1}{2}\,\sigma^2\,\partial_{xx}H^{(\lambda)} \sigma\,v\,\partial_{x\widehat{A}}H^{(\lambda)}  \hspace{10mm}\\
        -\, C_t\,x^2 + \frac{\lambda}{q-1} + \frac{(\gamma_t\,\partial_xH^{(\lambda)} + D_t\,x)^2}{4\,K_t} + \frac{q+1}{3\,q-1}\,\psi_t &= 0\,,\\
        H^{(\lambda)}(T,x,\widehat{A}) &= g(x)\,.
	\end{split}
\end{equation}
Given that $H^{(0)}$ is a classical solution to the HJB equation in \eqref{eqn:HJB_standard}, we see that a classical solution to \eqref{eqn:HJB_pf3} is obtained by
\begin{align}
    H^{(\lambda)}(t,x,\widehat{A}) &= H^{(0)}(t,x,\widehat{A}) - \int_t^T\alpha_q(u)\,du\,,
\end{align}
with $\alpha_q(t) = - \frac{\lambda}{q-1} - \frac{\psi_t\,(q+1)}{3\,q-1}$ as given in the statement of the proposition. Finally, the feedback form of the conditional density \eqref{eqn:cont_density_pf} results in the optimal control \eqref{eqn:pi_star_cont} for $\frac{1}{3} < q < 1$ as in the proposition. Other values of $q$ are handled similarly with minor modifications to the computations. \qed

\subsection{Proof of Lemma \ref{lem:convergence}}
\label{proof:convergence}

This proof is a combination of two results related to BSDEs. The first is that of approximating the solution to a BSDE with the solution to an appropriately chosen BS$\Delta$E. The second concerns the continuity of solutions to BSDEs with respect to a parameter. The condition
\begin{align*}
    \inf_{\substack{t\in[0,T] \\ \beta\in[0,\beta_0)}}\biggl\{K_t - h_t^{(2,\beta)}\,\gamma_t^2\,\beta\biggr\}>0\,,
\end{align*}
assumed in the statement of the proposition ensures that the BSDEs considered below have coefficients that guarantee solutions exists, in particular the coefficients have Lipschitz constants which can be chosen uniformly with respect to $\beta\in[0,\beta_0)$.

Fix $\epsilon > 0$. First, consider the BS$\Delta$E \eqref{eqn:BSDeltaE} (repeated here for convenience)
\begin{equation}
    \begin{split}
        h^{(1,N)}_n =&\; \E[h^{(1,N)}_{n+1}|\mcF_n] + 2\,h^{(2,N)}_{n+1}\,\widehat{A}_n\,\Delta t 	\\
        &\;+ \frac{(\E[h^{(1,N)}_{n+1}|\mcF_n]\,\gamma_n + 2\,h^{(2,N)}_{n+1}\,\widehat{A}_n\,\gamma_n\,\Delta t)\,(2\,h^{(2,N)}_{n+1}\,\gamma_n + D_n)}{2\,(K_n - h^{(2,N)}_{n+1}\,\gamma_n^2\,\Delta t)}\,\Delta t\,,\\
        h^{(2,N)}_n =&\; h^{(2,N)}_{n+1} - C_n\,\Delta t + \frac{(2\,h^{(2,N)}_{n+1}\,\gamma_n + D_n)^2}{4\,(K_n - h^{(2,N)}_{n+1}\,\gamma_n^2\,\Delta t)}\,\Delta t\,,
    \end{split}
\end{equation}
and note that occurrences of $\Delta t$ take two distinct roles. Some appearances act as a time step parameter as would appear in a finite-difference scheme, and others appear as a parameter within the coefficients. All occurrences of the second type are replaced below with a parameter $\beta\in[0,\beta_0)$ to generate the new parametrized BS$\Delta$E system
\begin{equation}
    \label{eqn:BSDeltaE_pf}
    \begin{split}
        h^{(1,N,\beta)}_n =&\; \E[h^{(1,N,\beta)}_{n+1}|\mcF_n] + 2\,h^{(2,N,\beta)}_{n+1}\,\widehat{A}_n\,\Delta t 	\\
        &\;+ \frac{(\E[h^{(1,N,\beta)}_{n+1}|\mcF_n]\,\gamma_n + 2\,h^{(2,N,\beta)}_{n+1}\,\widehat{A}_n\,\gamma_n\,\beta)\,(2\,h^{(2,N,\beta)}_{n+1}\,\gamma_n + D_n)}{2\,(K_n - h^{(2,N,\beta)}_{n+1}\,\gamma_n^2\,\beta)}\,\Delta t\,,	\\
        h^{(2,N,\beta)}_n =&\; h^{(2,N,\beta)}_{n+1} - C_n\,\Delta t + \frac{(2\,h^{(2,N,\beta)}_{n+1}\,\gamma_n + D_n)^2}{4\,(K_n - h^{(2,N,\beta)}_{n+1}\,\gamma_n^2\,\beta)}\,\Delta t\,,
    \end{split}
\end{equation}
and we note that $h^{(i,N,\Delta t)} = h^{(i,N)}$. Additionally, consider a new system of BSDEs (which is the continuous time analogue of \eqref{eqn:BSDeltaE_pf})
\begin{equation}\label{eqn:BSDE_beta_pf}
	\begin{split}
		dh^{(1,\beta)}_t &= -\biggl(2\,\widehat{A}_t\,h_t^{(2,\beta)} + \frac{(\gamma_t\,h_t^{(1,\beta)} + 2\,h^{(2,\beta)}_{t}\,\widehat{A}_t\,\gamma_t\,\beta)\,(2\,\gamma_t\,h_t^{(2,\beta)} + D_t)}{2\,(K_t - h^{(2,\beta)}_{t}\,\gamma_t^2\,\beta)}\biggr)\,dt + dZ_t\,,\\
		dh^{(2,\beta)}_t &= -\biggl(\frac{(2\,\gamma_t\,h_t^{(2,\beta)} + D_t)^2}{4\,(K_t - h^{(2,\beta)}_{t}\,\gamma_t^2\,\beta)} - C_t\biggr)\,dt\,.
	\end{split}
\end{equation}
From Theorem 6 in \cite{gobet2007error}, there exists a sufficiently large $N^*$, which can be chosen uniformly in $\beta\in[0,\beta_0)$, such that for $N\geq N^*$
\begin{align*}
    \max_{0\leq n\leq N}\mathbb{E}\biggl[(h^{(1,\beta)}_{t_n} - h^{(1,N,\beta)}_n)^2\biggr] + \max_{0\leq n\leq N}(h^{(2,\beta)}_{t_n} - h^{(2,N,\beta)}_n)^2 & < \epsilon/2\,.
\end{align*}
Next, treating $h^{(i,\beta)}$ as a family of solutions to \eqref{eqn:BSDE_beta_pf} that depend on the parameter $\beta$, from \cite{el1997backward}, there exists sufficiently small $\beta^*$ such that
\begin{align*}
    & \hspace*{-2em} \max_{0\leq n\leq N}\mathbb{E}\biggl[(h^{(1,\beta^*)}_{t_n} - h^{(1,0)}_{t_n})^2\biggr] + \max_{0\leq n\leq N}(h^{(2,\beta^*)}_{t_n} - h^{(2,0)}_{t_n})^2\\
    \leq & \sup_{0\leq t\leq T}\mathbb{E}\biggl[(h^{(1,\beta^*)}_t - h^{(1,0)}_t)^2\biggr] + \sup_{0\leq t\leq T}(h^{(2,\beta^*)}_t - h^{(2,0)}_t)^2\\
    \leq &\; \mathbb{E}\biggl[\sup_{0\leq t\leq T}(h^{(1,\beta^*)}_t - h^{(1,0)}_t)^2\biggr] + \sup_{0\leq t\leq T}(h^{(2,\beta^*)}_t - h^{(2,0)}_t)^2\\
    \leq &\; \epsilon/2\,.
\end{align*}
Thus, by choosing $\Delta t < \min(\beta^*,T/N^*)$ and combining the previous inequalities, we have
\begin{align*}
    \max_{0\leq n\leq N}\mathbb{E}\biggl[(h^{(1,0)}_{t_n} - h^{(1,N,\Delta t)}_n)^2\biggr] + \max_{0\leq n\leq N}(h^{(2,0)}_{t_n} - h^{(2,N,\Delta t)}_n)^2 & < \epsilon\,,
\end{align*}
which is the desired result.
\qed

\subsection{Derivation of $\pi^*(\nu|x)$ in \eqref{eqn:pi_star_Q}}
\label{proof:pi_star_Q}

We present the necessary computations for $q>1$. Other cases follow similarly. For fixed $x\in\mathcal{X}$, the quantity to be maximized is
\begin{align*}
    Q(x;\pi) &= \sum_{\nu\in\mcA} Q(x;\delta_\nu)\,\pi(\nu|x) + \frac{\lambda}{q-1}\biggl(1 - \sum_{\nu\in\mcA} (\pi(\nu|x))^q\biggr)\,,
\end{align*}
subject to $\sum_{\nu\in\mathcal{A}} \pi(\nu|x) = 1$ and $\pi(\nu|x)\geq 0$. This is a finite dimensional constrained maximization with respect to the quantity $\pi(\nu|x)$. A Lagrange multiplier $\psi(x)$ is introduced for the equality constraint and for each $\nu$ a KKT multiplier $\zeta(\nu;x)$ is introduced for the inequality constraint. The corresponding Lagrangian is
\begin{align*}
    \mathcal{L}(\pi(\nu|x);x,\psi(x)) &= \sum_{\nu\in\mcA} Q(x;\delta_\nu)\,\pi(\nu|x) + \frac{\lambda}{q-1}\biggl(1 - \sum_{\nu\in\mcA} (\pi(\nu|x))^q\biggr)\\ 
    & \hspace{15mm} + \psi(x)\biggl(\sum_{\nu\in\mathcal{A}} \pi(\nu|x)-1\biggr) + \sum_{\nu\in\mathcal{A}}\zeta(\nu;x)\,\pi(\nu|x)\,.
\end{align*}
First order conditions with respect to $\pi(\nu|x)$ imply
\begin{align*}
    Q(x;\delta_\nu) - \frac{\lambda\,q}{q-1}(\pi(\nu|x))^{q-1} + \psi(x) + \zeta(\nu;x) &= 0\,,
\end{align*}
for each $\nu\in\mathcal{A}$. This is rewritten as
\begin{align}
    \pi(\nu|x) &= \biggl(\frac{q-1}{\lambda\,q}\biggr)^{\frac{1}{q-1}}\biggl(\zeta(\nu;x) + \psi(x) + Q(x;\delta_\nu)\biggr)^{\frac{1}{q-1}}\,.\label{eqn:Q_pi_pf1}
\end{align}
The conditions $\zeta(\nu;x)\geq 0$, $\pi(\nu|x)\geq 0$, and $\zeta(\nu;x)\,\pi(\nu|x)=0$ imply
\begin{align*}
    \zeta(\nu;x) &= [-\psi(x) - Q(x;\delta_\nu)]_+\,,
\end{align*}
which when substituted back into \eqref{eqn:Q_pi_pf1} give
\begin{align*}
    \pi^*(\nu|x) &= \biggl(\frac{q-1}{\lambda\,q}\biggr)^{\frac{1}{q-1}}\biggl[\psi(x) + Q(x;\delta_\nu)\biggr]_+^{\frac{1}{q-1}}\,.
\end{align*}
Finally, the quantity $\psi(x)$ is chosen so that $\sum_{\nu\in\mathcal{A}} \pi(\nu|x) = 1$. This is always possible because the expression above is continuous and increasing with respect to $\psi(x)$, and attains all values from $0$ to infinity.

 \section*{References}

\bibliographystyle{chicago}
\bibliography{References}

\begin{thebibliography}{}

\bibitem[\protect\citeauthoryear{El~Karoui, Peng, and Quenez}{El~Karoui et~al.}{1997}]{el1997backward}
El~Karoui, N., S.~Peng, and M.~C. Quenez (1997).
\newblock Backward stochastic differential equations in finance.
\newblock {\em Mathematical finance\/}~{\em 7\/}(1), 1--71.

\bibitem[\protect\citeauthoryear{Firoozi and Jaimungal}{Firoozi and Jaimungal}{2022}]{firoozi2022exploratory}
Firoozi, D. and S.~Jaimungal (2022).
\newblock Exploratory lqg mean field games with entropy regularization.
\newblock {\em Automatica\/}~{\em 139}, 110177.

\bibitem[\protect\citeauthoryear{Fox, Pakman, and Tishby}{Fox et~al.}{2015}]{fox2015taming}
Fox, R., A.~Pakman, and N.~Tishby (2015).
\newblock Taming the noise in reinforcement learning via soft updates.
\newblock {\em arXiv preprint arXiv:1512.08562\/}.

\bibitem[\protect\citeauthoryear{Fujisaki, Kallianpur, and Kunita}{Fujisaki et~al.}{1972}]{fujisaki1972stochastic}
Fujisaki, M., G.~Kallianpur, and H.~Kunita (1972).
\newblock Stochastic differential equations for the non linear filtering problem.
\newblock {\em Osaka Journal of Mathematics\/}~{\em 9\/}(1), 19--40.

\bibitem[\protect\citeauthoryear{Gao, Xu, and Zhou}{Gao et~al.}{2022}]{gao2022state}
Gao, X., Z.~Q. Xu, and X.~Y. Zhou (2022).
\newblock State-dependent temperature control for langevin diffusions.
\newblock {\em SIAM Journal on Control and Optimization\/}~{\em 60\/}(3), 1250--1268.

\bibitem[\protect\citeauthoryear{Geist, Scherrer, and Pietquin}{Geist et~al.}{2019}]{geist2019theory}
Geist, M., B.~Scherrer, and O.~Pietquin (2019).
\newblock A theory of regularized markov decision processes.
\newblock In {\em International Conference on Machine Learning}, pp.\  2160--2169. PMLR.

\bibitem[\protect\citeauthoryear{Gobet and Labart}{Gobet and Labart}{2007}]{gobet2007error}
Gobet, E. and C.~Labart (2007).
\newblock Error expansion for the discretization of backward stochastic differential equations.
\newblock {\em Stochastic processes and their applications\/}~{\em 117\/}(7), 803--829.

\bibitem[\protect\citeauthoryear{Gradshteyn and Ryzhik}{Gradshteyn and Ryzhik}{2014}]{gradshteyn2014table}
Gradshteyn, I.~S. and I.~M. Ryzhik (2014).
\newblock {\em Table of integrals, series, and products}.
\newblock Academic press.

\bibitem[\protect\citeauthoryear{Guan, Zhang, and Tsiotras}{Guan et~al.}{2020}]{guan2020learning}
Guan, Y., Q.~Zhang, and P.~Tsiotras (2020).
\newblock Learning nash equilibria in zero-sum stochastic games via entropy-regularized policy approximation.
\newblock {\em arXiv preprint arXiv:2009.00162\/}.

\bibitem[\protect\citeauthoryear{Guo, Xu, and Zariphopoulou}{Guo et~al.}{2022}]{guo2022entropy}
Guo, X., R.~Xu, and T.~Zariphopoulou (2022).
\newblock Entropy regularization for mean field games with learning.
\newblock {\em Mathematics of Operations Research\/}.

\bibitem[\protect\citeauthoryear{Hao, Zhang, Shi, and Li}{Hao et~al.}{2022}]{hao2022entropy}
Hao, D., D.~Zhang, Q.~Shi, and K.~Li (2022).
\newblock Entropy regularized actor-critic based multi-agent deep reinforcement learning for stochastic games.
\newblock {\em Information Sciences\/}.

\bibitem[\protect\citeauthoryear{Jaimungal}{Jaimungal}{2022}]{jaimungal2022reinforcement}
Jaimungal, S. (2022).
\newblock Reinforcement learning and stochastic optimisation.
\newblock {\em Finance and Stochastics\/}~{\em 26\/}(1), 103--129.

\bibitem[\protect\citeauthoryear{Liptser and Shiryaev}{Liptser and Shiryaev}{2013}]{liptser2013statistics2}
Liptser, R. and A.~N. Shiryaev (2013).
\newblock {\em Statistics of Random Processes: II. Applications}, Volume~6.
\newblock Springer Science \& Business Media.

\bibitem[\protect\citeauthoryear{Neu, Jonsson, and G{\'o}mez}{Neu et~al.}{2017}]{neu2017unified}
Neu, G., A.~Jonsson, and V.~G{\'o}mez (2017).
\newblock A unified view of entropy-regularized markov decision processes.
\newblock {\em arXiv preprint arXiv:1705.07798\/}.

\bibitem[\protect\citeauthoryear{Savas, Ahmadi, Tanaka, and Topcu}{Savas et~al.}{2019}]{savas2019entropy}
Savas, Y., M.~Ahmadi, T.~Tanaka, and U.~Topcu (2019).
\newblock Entropy-regularized stochastic games.
\newblock In {\em 2019 IEEE 58th Conference on Decision and Control (CDC)}, pp.\  5955--5962. IEEE.

\bibitem[\protect\citeauthoryear{Tang, Zhang, and Zhou}{Tang et~al.}{2022}]{tang2022exploratory}
Tang, W., Y.~P. Zhang, and X.~Y. Zhou (2022).
\newblock Exploratory hjb equations and their convergence.
\newblock {\em SIAM Journal on Control and Optimization\/}~{\em 60\/}(6), 3191--3216.

\bibitem[\protect\citeauthoryear{Thistleton, Marsh, Nelson, and Tsallis}{Thistleton et~al.}{2007}]{thistleton2007generalized}
Thistleton, W.~J., J.~A. Marsh, K.~Nelson, and C.~Tsallis (2007).
\newblock Generalized box--m{\"u}ller method for generating $ q $-gaussian random deviates.
\newblock {\em IEEE transactions on information theory\/}~{\em 53\/}(12), 4805--4810.

\bibitem[\protect\citeauthoryear{Todorov and Li}{Todorov and Li}{2005}]{todorov2005generalized}
Todorov, E. and W.~Li (2005).
\newblock A generalized iterative lqg method for locally-optimal feedback control of constrained nonlinear stochastic systems.
\newblock In {\em Proceedings of the 2005, American Control Conference, 2005.}, pp.\  300--306. IEEE.

\bibitem[\protect\citeauthoryear{Tsallis}{Tsallis}{1988}]{tsallis1988possible}
Tsallis, C. (1988).
\newblock Possible generalization of boltzmann-gibbs statistics.
\newblock {\em Journal of statistical physics\/}~{\em 52\/}(1), 479--487.

\bibitem[\protect\citeauthoryear{Tsallis}{Tsallis}{2011}]{tsallis2011nonadditive}
Tsallis, C. (2011).
\newblock The nonadditive entropy sq and its applications in physics and elsewhere: Some remarks.
\newblock {\em Entropy\/}~{\em 13\/}(10), 1765--1804.

\bibitem[\protect\citeauthoryear{Wang, Zariphopoulou, and Zhou}{Wang et~al.}{2020}]{wang2020reinforcement}
Wang, H., T.~Zariphopoulou, and X.~Y. Zhou (2020).
\newblock Reinforcement learning in continuous time and space: A stochastic control approach.
\newblock {\em J. Mach. Learn. Res.\/}~{\em 21\/}(198), 1--34.

\bibitem[\protect\citeauthoryear{Ziebart}{Ziebart}{2010}]{ziebart2010modeling}
Ziebart, B.~D. (2010).
\newblock {\em Modeling purposeful adaptive behavior with the principle of maximum causal entropy}.
\newblock Carnegie Mellon University.

\end{thebibliography}

\end{document}